\documentclass[journal,draftclsnofoot,onecolumn,12pt,twoside]{IEEEtran}

\normalsize

\usepackage[T1]{fontenc}% optional T1 font encoding
\usepackage{amsfonts}
\usepackage{amsmath}
\usepackage{amsthm}
\usepackage{amssymb}
\usepackage{caption}
\usepackage[labelformat=simple]{subcaption}
\usepackage{cite,graphicx,algorithm}
\usepackage{algorithmic}
\usepackage{color}
\usepackage{bm}

\ifCLASSINFOpdf
\else
\fi

\theoremstyle{plain}

\newtheorem{proposition}{Proposition}
\newtheorem{remark}{Remark}

\newcommand*{\LargerCdot}{\raisebox{-0.25ex}{\scalebox{2}{$\cdot$}}}

\begin{document}
\title{OTFS Signaling for Uplink NOMA of Heterogeneous Mobility Users}

\author{Yao~Ge,~\IEEEmembership{Student Member,~IEEE,}
        Qinwen~Deng,~\IEEEmembership{Student Member,~IEEE,}\\
        P.~C.~Ching,~\IEEEmembership{Fellow,~IEEE,}
        and~Zhi~Ding,~\IEEEmembership{Fellow,~IEEE}% <-this % stops a space
\thanks{Y. Ge and P. C. Ching are with the Department
of Electronic Engineering, The Chinese University of Hong Kong, Hong Kong SAR of China (e-mail: yaoge.gy.jay@hotmail.com; pcching@ee.cuhk.edu.hk).}% <-this % stops a space
\thanks{Q. Deng and Z. Ding are with the Department of Electrical and Computer Engineering, University of California at Davis, Davis, CA 95616 USA (e-mail: mrdeng@ucdavis.edu; zding@ucdavis.edu).}}% <-this % stops a space
%\thanks{TCOM version based on Michael Shell's bare{\textunderscore}jrnl.tex version 1.3.}}


\markboth{}%
{}



\maketitle

\begin{abstract}
%\boldmath
We investigate a coded uplink non-orthogonal multiple access (NOMA)
configuration in which groups of co-channel users 
are modulated in accordance with orthogonal time frequency space (OTFS).
We take advantage of OTFS characteristics to achieve
NOMA spectrum sharing in the 
delay-Doppler domain between 
stationary and mobile users.
We develop an efficient iterative 
turbo receiver based on the principle of successive interference 
cancellation (SIC) to overcome the co-channel interference (CCI).
We propose two turbo detector algorithms: 
orthogonal approximate message passing with 
linear minimum mean squared error (OAMP-LMMSE) and 
Gaussian approximate message passing with expectation propagation (GAMP-EP).
The interactive OAMP-LMMSE detector and GAMP-EP detector
are respectively assigned for the reception of the stationary and mobile users.
We analyze the convergence performance of our proposed iterative 
SIC turbo receiver by utilizing a customized  
extrinsic information transfer (EXIT) chart and simplify the corresponding detector algorithms to further reduce receiver complexity. 
Our proposed iterative SIC turbo receiver demonstrates 
performance improvement over existing receivers and 
robustness against imperfect SIC process and channel state information uncertainty.
\end{abstract}

\begin{IEEEkeywords}
Co-channel interference, NOMA, OTFS, Resource allocation, Turbo receiver.
\end{IEEEkeywords}

% \fi
\IEEEpeerreviewmaketitle

\section{Introduction}
% \IEEEPARstart{T}{his} demo file

The rapid advancement of wireless transmission and mobile communication 
techniques has led to an explosive rise in data traffics of wireless networks.
In order to support such tremendous needs, 
high spectrum efficiency techniques such as
non-orthogonal multiple access (NOMA) 
have been considered as promising solutions
for improving spectrum utilization and user connectivity \cite{ding2017survey,ding2017application,dai2015non}. 
Unlike conventional orthogonal multiple access (OMA), 
NOMA allows multiple users to access the same spectrum
simultaneously at different power levels \cite{islam2016power} 
or with the help of low-density spreading codes \cite{sharma2019joint}. 
To overcome the inevitable co-channel interference (CCI),
advanced receivers such as successive interference cancellation (SIC) 
are required as effective multi-user detection for NOMA systems. 
Note that the NOMA spectrum sharing is common among users with different
channel conditions or quality of service (QoS) levels \cite{ding2017application,islam2016power}.

Broadband mobile communications in high-mobility environments 
such as high-speed railways and autonomous vehicles represent another
strong arena of growth. High mobility communications are particularly
challenging because of the well-known fast channel fading 
and distortions due to the large Doppler spread. 
Recently, the advent of orthogonal time frequency space 
(OTFS) \cite{hadani2017orthogonal} modulation shows strong promise
as an effective PHY-layer alternative to traditional orthogonal frequency division multiplexing (OFDM) in high-mobility environments.
OTFS exhibits performance advantages since it can exploit the diversity gain coming from both the delay and Doppler channel domains. 
OTFS can effectively convert a rapidly time-varying channel in time-frequency domain
into a quasi-stationary channel model in delay-Doppler domain. 
This quasi-stationary channel model simplifies channel estimation \cite{shen2019channel,raviteja2019embedded,9110823} and symbol detection \cite{surabhi2019low,murali2018otfs,tiwari2019low,raviteja2018interference,yuan2020simple} for wireless receivers in high-mobility scenarios. 
OTFS achieves diversity gain for stationary multipath 
channels and doubly-selective channels, as shown in \cite{raviteja2019otfs} 
and \cite{surabhi2019diversity,raviteja2019effective}, respectively. 
Other related works considered OTFS modulation 
in multiple-input multiple-output (MIMO) systems
\cite{ramachandran2018mimo}
and millimeter wave (mmWave) 
communication systems \cite{surabhi2019otfs}.

Recognizing the superior performance of OTFS in delay-Doppler channel domain, 
several works studied multiple user access based on
OTFS framework in high-mobility scenarios \cite{khammammetti2018otfs,augustine2019interleaved,surabhi2019multiple,li2020new,chatterjee2020non,deka2020otfs}. 
For OTFS-OMA, the authors in \cite{khammammetti2018otfs,augustine2019interleaved} proposed to allocate different time-frequency resources to different users. 
This orthogonal resource allocation can be either in contiguous \cite{khammammetti2018otfs} or interleaved \cite{augustine2019interleaved} fashions when the ideal bi-orthogonal pulses are available. Practically, however, such ideal pulses are not realizable in view of
the Heisenberg uncertainty principle \cite{matz2013time}. 
Allocation of different delay-Doppler resources to different users 
\cite{surabhi2019multiple} has shown significant CCI that requires complex receivers.
For massive MIMO-OTFS networks, a new path division multiple access (PDMA)
\cite{li2020new} can assign angle-domain resources 
to different users so as to eliminate CCI at the receiver 
in the angle-delay-Doppler domain. 
To further improve the spectral efficiency and support the massive connectivity, the OTFS-NOMA schemes were proposed in \cite{chatterjee2020non,deka2020otfs}, where multiple mobile users are sharing the same delay-Doppler resources and distinguished by either different power levels \cite{chatterjee2020non} or sparse codewords \cite{deka2020otfs}.

A recent work \cite{ding2019otfs,ding2019robust}
suggested a new form of OTFS-NOMA in which a single high-mobility user using 
OTFS in the delay-Doppler domain is paired with a group of
low-mobility OFDM users for non-orthogonal channel sharing. 
The improved spectrum efficiency shown by \cite{ding2019otfs,ding2019robust},
however, relies on the ideal bi-orthogonal OTFS pulses have been elusive
to practitioners and may not even exist physically. 
These ideal bi-orthogonal OTFS pulses are also essential to the proposed simple equalizations in \cite{ding2019otfs,ding2019robust}
for OTFS-NOMA system. In addition, the performance analysis of OTFS-NOMA scheme in \cite{ding2019otfs,ding2019robust} only considers the perfect SIC process and requires mobile channels to exhibit
on-the-grid delays and Doppler shifts, which are still unrealistic assumptions in practical OTFS-NOMA system deployment.

To alleviate the dependency on the assumptions of ideal bi-orthogonal OTFS pulses and
on-the-grid channel delays and/or Doppler shifts, 
we investigate a more general coded uplink OTFS-NOMA scenario in this work. 
Without loss of generality, we focus on a simple scenario
in which
the stationary users and mobile users are grouped for NOMA. 
Unlike \cite{ding2019otfs,ding2019robust}, both users utilize OTFS 
modulation. 
We design an efficient NOMA protocol by grouping users with 
different mobility profiles. With the use of
OTFS, co-channel users of different mobility profiles
can take advantage of different resource allocations in the delay-Doppler domain
to mitigate their CCI and simplify the receiver complexity. 
We eliminate the unrealistic assumptions of ideal bi-orthogonal
pulses, perfect SIC process 
and on-the-grid channel delay/Doppler shifts. 
By using the practical pulses such as rectangular pulses,
the block circulant matrices in \cite{ding2019otfs,ding2019robust} 
no longer apply. To this end, we develop a novel
receiver architecture to 
effectively mitigate CCI and recover the signal for each user.
Our contributions in this work are summarized as follows:
\begin{enumerate}
\item We propose an OTFS-based NOMA (OBNOMA) configuration that groups users of 
different mobility profiles by only utilizing OTFS
modulation. Without loss of generality, we consider groups
of stationary and mobile users that occupy different
sub-vector resources of delay-Doppler domain via OTFS. 
This OBNOMA framework can effectively tackle the CCI and
is amenable to effective receiver algorithms. 

\item We design an iterative turbo receiver for multi-user 
detection and decoding that leverages the SIC principle. 
The proposed joint SIC detector and individual user decoders
exchange the extrinsic information iteratively 
to improve receiver performance.  
In particular, we develop an orthogonal approximate message passing
with linear minimum mean squared error (OAMP-LMMSE) algorithm for detecting
the OBNOMA stationary users' signal and 
a Gaussian approximate message passing
with expectation propagation (GAMP-EP) algorithm for detecting 
the signal of OBNOMA mobile users. 

\item We develop a novel customized extrinsic information transfer (EXIT) chart framework to analyze the convergence property of our proposed iterative SIC turbo receiver. More importantly, we further propose 
reduced complexity variants for both
OAMP-LMMSE and GAMP-EP detectors without significant performance drop.

\item We demonstrate that our proposed
iterative SIC turbo receiver for OBNOMA system 
outperforms the existing methods and has robustness to the imperfect SIC process and channel state information (CSI) uncertainty. 
\end{enumerate}

We organize the remaining sections of 
the paper as follows. Section
\ref{II_OTFS} summarizes the basics of OTFS transmission. 
Section \ref{III_Model} proposes the novel coded uplink OBNOMA system model
and describes the resource allocation in delay-Doppler domain 
for OBNOMA users. In Section \ref{IV_Receiver}, we propose
our iterative SIC turbo receiver and the two component
detectors respectively for signal detection of stationary 
and mobile users. We further analyze the convergence behavior
of the proposed iterative SIC turbo receiver through a novel customized EXIT chart
and simplify the two
detector algorithms to reduce receiver complexity in Section \ref{V_Reduced}. 
We provide simulation results in Section \ref{VI_Simulation} 
under different scenarios. Our conclusions are finally drawn in Section \ref{VII_Conclusion}. The Appendix contains some detailed proofs at the end of the paper.

%Such spectrum sharing among the users with
%different mobility profiles can be particularly important to
%5G and beyond communication scenarios, where some users
%might be static, e.g., Internet of Things (IoT) sensors, and there might be some users which are moving at very high
%speeds, e.g., users in a high-speed train.
%
%Such spectrum sharing is particularly important if
%the high-mobility user has weak channel conditions or needs
%to be served with a small data rate only.

\section{Preliminaries}\label{II_OTFS}
In this section, we briefly introduce basic OTFS concepts and system transmission model. We also
provide the mathematical descriptions of OTFS used for both mobile and stationary
users.

\subsection{Basic Concepts of OTFS}
Unlike conventional OFDM, OTFS multiplexes and processes each information symbol in the roughly constant delay-Doppler domain rather than time-frequency domain. A lattice in 
time-frequency plane is sampled by intervals
$T$ (seconds) and
$\Delta f{\rm{ = }}{1 \mathord{\left/
 {\vphantom {1 T}} \right.
 \kern-\nulldelimiterspace} T}$
(Hz) along the time and frequency axes, i.e., $$\Lambda  = \left\{ {(m\Delta f,nT),m = 0, \cdots ,M - 1;n = 0, \cdots ,N - 1} \right\},
$$
where $M \in {\cal Z}$ and $N \in {\cal Z}$ represent the total available numbers of subcarriers and time intervals, respectively. According to the channel characteristics, $T$ and ${\Delta f}$ are chosen, respectively, larger than the maximal channel delay spread and maximum Doppler frequency shift. 

The corresponding lattice in delay-Doppler plane is described by $$\Gamma  = \left\{ {\left(\frac{\ell}{{M\Delta f}},\frac{k}{{NT}}\right),\ell = 0, \cdots ,M - 1;k = 0, \cdots ,N - 1} \right\},$$
where ${1 \mathord{\left/
{\vphantom {1 {M\Delta f}}} \right.
 \kern-\nulldelimiterspace} {M\Delta f}}$ and ${1 \mathord{\left/
 {\vphantom {1 {NT}}} \right.
 \kern-\nulldelimiterspace} {NT}}$ denote the resolutions of delay dimension and Doppler dimension, respectively. Note that signals placed on delay-Doppler grids in a given packet burst can be transformed into time-frequency samples with time duration ${T_f} = NT$ and bandwidth $B = M\Delta f$. Additional details on OTFS can be found in existing works such as \cite{hadani2017orthogonal,raviteja2018interference}.

\subsection{OTFS Signal Model for Mobile Users}
At OTFS transmitter, the $MN$ random information symbols (e.g., QAM) are generated from a complex alphabet $\mathbb{A} = \left\{ {{a_1},{a_2}, \cdots ,{a_Q}} \right\}$ and placed on the delay-Doppler plane $\Gamma $. These delay-Doppler symbols ${\bf{X}} \in {\mathbb{C}^{M \times N}}$ are mapped into a lattice in time-frequency domain ${{\bf{\bar X}}}\in {\mathbb{C}^{M \times N}}$ through the inverse symplectic finite Fourier transform (ISFFT) \cite{raviteja2018practical}, 
\begin{align}\label{ISFFT}
{\bf{\bar X}} = {{\bf{F}}_M}{\bf{XF}}_N^H,
\end{align}
where ${{\bf{F}}_M} \in {\mathbb{C}^{M \times M}}$ and ${{\bf{F}}_N} \in {\mathbb{C}^{N \times N}}$ are the normalized $M$-point and $N$-point discrete Fourier transform (DFT) matrices, respectively.
Next, the Heisenberg transform is adopted to the time-frequency signal ${{\bf{\bar X}}}$ with a transmit pulse ${g_{tx}}(t)$ to generate the time domain signal ${\bf{s}} \in {\mathbb{C}^{MN \times 1}}$,
\begin{align}
s[c] = \sum\limits_{n = 0}^{N - 1} {\sum\limits_{m = 0}^{M - 1} {{\bar X}[m,n]{g_{tx}}(c{T_s} - nT){e^{j2\pi m \Delta f(c{T_s} - nT)}}} },\; c = 0, \cdots, MN - 1,
\end{align}
where ${T_s} = {1 \mathord{\left/
 {\vphantom {1 {M\Delta f}}} \right.
 \kern-\nulldelimiterspace} {M\Delta f}}$ is the symbol spaced sampling interval.
 
To overcome the inter-frame interference, we append a cyclic prefix (CP) of length no shorter than 
the maximal channel delay spread to signal ${\bf{s}}$. After passing a transmit filter, the resulted time domain signal enters the multipath fading channels characterized by sampled response of
\begin{align}\label{channel_M}
h[c,p] = \sum\limits_{i = 1}^L {{h_i}{e^{j2\pi {\nu _i}\left( {c{T_s} - p{T_s}} \right)}}{{\mathop{\rm P}\nolimits} _\text{{rc}}}(p{T_s} - {\tau _i})}, \; c = 0, \cdots, MN - 1;\; p = 0, \cdots ,P - 1,
\end{align}
where $L$ denotes the number of multipaths;
${{h_i}}$, ${{\tau _i}}$ and ${{\nu _i}}$ are
the complex gain, delay and Doppler shift associated with the $i$-th path, respectively. The channel tap $P$ is determined by the maximal channel delay spread as well as the duration of the overall filter response.

In (\ref{channel_M}),
${{{\mathop{\rm P}\nolimits} _\text{{rc}}}(p{T_s} - {\tau _i})}$ is the sampled overall
filter response that
comprises a pair of bandlimiting matched filters adopted by the transmitter and receiver to control signal transmission bandwidth and to achieve maximum
signal-to-noise ratio (SNR) at the receiver.
In practice, the most common
implemented pulse shaping filters at the transmitter and receiver are the root raised-cosine (RRC) filters, leading to a raised-cosine (RC) rolloff
pulse for
${{{\mathop{\rm P}\nolimits} _\text{{rc}}}(\tau )}$. 
In addition, the Doppler frequency shift of the $i$-th path can be written as ${\nu _i} = ({{k_{{\nu _i}}} + {\beta _{{\nu _i}}}})/NT$,
where the integer ${{k_{{\nu _i}}}}$ and real
${\beta _{{\nu _i}}} \in \left( { -0.5,0.5} \right]$ are respectively represent the index and fractional Doppler shift of ${\nu _i}$.

Consider the baseband model. At OTFS receiver, the channel output 
signal enters a user-defined receive filter. After removing CP, we can obtain
the received signal ${\bf{r}} \in {\mathbb{C}^{MN \times 1}}$ as
\begin{align}
r[c] = \sum\limits_{p = 0}^{P - 1} {h[c,p]s\left[ {{{\left[ {c - p} \right]}_{MN}}} \right]}  + n[c],\; c = 0, \cdots ,MN - 1,
\end{align}
where $\bf{n}$ represents the filtered noise and 
the notation ${\left[  \LargerCdot  \right]_m}$ denotes mod-$m$ operation.
The received time domain signal ${\bf{r}}$ is then processed by Wigner transform (i.e., the inverse of Heisenberg transform) using a receive pulse ${g_{rx}}(t)$ to produce the time-frequency domain signal
\begin{align}
\bar Y[m,n] = \sum\limits_{c = 0}^{MN - 1} {g_{rx}^*(c{T_s} - nT)r[c]{e^{ - j2\pi m\Delta f(c{T_s} - nT)}}}, \; {m = 0, \cdots ,M - 1};\; {n = 0, \cdots ,N - 1}.
\end{align}
Finally, the signal matrix ${\bf{\bar Y}} \in {\mathbb{C}^{M \times N}}$ in the time-frequency domain are transformed back to the delay-Doppler domain via symplectic finite Fourier transform (SFFT) as described below \cite{raviteja2018practical}:
\begin{align}
{\bf{Y}} = {\bf{F}}_M^H{\bf{\bar Y}}{{\bf{F}}_N}.
\end{align}

For simplicity, we utilize a rectangular pulse 
for ${g_{tx}}(t)$ and ${g_{rx}}(t)$ in the above steps, for which
the baseband OTFS input-output relationship
in delay-Doppler domain is given by \cite{ge2020}
\begin{align}\label{relate_Mobility}
Y[\ell,k] = \sum\limits_{p = 0}^{P - 1} {\sum\limits_{i = 1}^L {\sum\limits_{q = 0}^{N - 1} {{h_i}{{\mathop{\rm P}\nolimits} _\text{{rc}}}(p{T_s} - {\tau _i})\gamma (k,\ell,p,q,{k_{{\nu _i}}},{\beta _{{\nu _i}}})X\left[ {{{\left[ {\ell - p} \right]}_M},{{\left[ {k - {k_{{\nu _i}}} + q} \right]}_N}} \right]} } }  + \omega [\ell,k],
\end{align}
%\begin{subequations}\label{relate_Mobility}
%\begin{align}
%Y[\ell,k] &= \sum\limits_{p = 0}^{P - 1} {\sum\limits_{i = 1}^L {\sum\limits_{q = 0}^{N - 1} {{h_i}{{\mathop{\rm P}\nolimits} _\text{{rc}}}(p{T_s} - {\tau _i})\gamma (k,\ell,p,q,{k_{{\nu _i}}},{\beta _{{\nu _i}}})X\left[ {{{\left[ {\ell - p} \right]}_M},{{\left[ {k - {k_{{\nu _i}}} + q} \right]}_N}} \right]} } }  + \omega [\ell,k]\\
%&\approx  \sum\limits_{p = 0}^{P - 1} {\sum\limits_{i = 1}^L {\sum\limits_{q = - {E_i}}^{{E_i}} {{h_i}{{\mathop{\rm P}\nolimits} _\text{{rc}}}(p{T_s} - {\tau _i})\gamma (k,\ell,p,q,{k_{{\nu _i}}},{\beta _{{\nu _i}}})X\left[ {{{\left[ {\ell - p} \right]}_M},{{\left[ {k - {k_{{\nu _i}}} + q} \right]}_N}} \right]} } }  + \omega [\ell,k],
%\end{align}
%\end{subequations}
where ${{\bm{\omega }}} \in {\mathbb{C}^{M \times N}} $ is the noise at the output of the SFFT. We also use the following definitions:
\begin{subequations}
\begin{equation}
\gamma (k,\ell,p,q,{k_{{\nu _i}}},{\beta _{{\nu _i}}}) =
\begin{cases}
\frac{1}{N}\xi (\ell,p,{k_{{\nu _i}}},{\beta _{{\nu _i}}})\theta (q,{\beta _{{\nu _i}}}),&p \le \ell < M,\\
\frac{1}{N}\xi (\ell,p,{k_{{\nu _i}}},{\beta _{{\nu _i}}})\theta (q,{\beta _{{\nu _i}}})\phi (k,q,{k_{{\nu _i}}}), &0 \le \ell < p,
\end{cases}
\end{equation}
\begin{align}
\xi (\ell,p,{k_{{\nu _i}}},{\beta _{{\nu _i}}}) = {e^{j2\pi \left( {\frac{{\ell - p}}{M}} \right)\left( {\frac{{{k_{{\nu _i}}} + {\beta _{{\nu _i}}}}}{N}} \right)}},
\end{align}
\begin{align}
\theta (q,{\beta _{{\nu _i}}}) = \frac{{{e^{ - j2\pi ( - q - {\beta _{{\nu _i}}})}} - 1}}{{{e^{ - j\frac{{2\pi }}{N}( - q - {\beta _{{\nu _i}}})}} - 1}},
\;
\phi (k,q,{k_{{\nu _i}}}) = {e^{ - j2\pi \frac{{{{\left[ {k - {k_{{\nu _i}}} + q} \right]}_N}}}{N}}}.
\end{align}
\end{subequations}

To summarize,
the input-output relationship in (\ref{relate_Mobility}) can be
vectorized column-wise into
\begin{align}\label{input_output_M}
{{{\bf{\tilde y}}}_{\cal M}} = {{{\bf{\tilde H}}}_{\cal M}}{{{\bf{\tilde x}}}_{\cal M}} + {{\bm{\omega }}_\mathcal{M}},
\end{align}
where ${{{\bf{\tilde x}}}_{\cal M}},{{{\bf{\tilde y}}}_{\cal M}},{{\bm{\omega }}_\mathcal{M}} \in {\mathbb{C}^{MN \times 1}} $ and ${{{\bf{\tilde H}}}_{\cal M}}\in {\mathbb{C}^{MN \times MN}}$ is a sparse matrix.

\subsection{OTFS Signal Model for Stationary Users}
Since the stationary user do not experience Doppler shifts (i.e., ${\nu _i} = 0,\forall i$), the baseband channel impulse response in (\ref{channel_M}) can be simplified into
\begin{align}
h[p] = \sum\limits_{i = 1}^L {{h_i}{{\mathop{\rm P}\nolimits} _\text{{rc}}}(p{T_s} - {\tau _i})}, \;  p = 0, \cdots ,P - 1.
\end{align}
In this case, the channel input-output relationship of OTFS in (\ref{relate_Mobility}) reduces to
\begin{align}\label{relate_Static}
Y[\ell,k] = \sum\limits_{p = 0}^{P - 1}  {h[p]\bar \gamma (k,\ell ,p)X\left[ {{{\left[ {\ell - p} \right]}_M},k} \right]}  + \omega [\ell,k],
\end{align}
where
\begin{equation}\label{phase_static}
\bar \gamma (k,\ell ,p) =
\begin{cases}
1,&p \le \ell < M,\\
{e^{ - j2\pi \frac{k}{N}}}, &0 \le \ell < p.
\end{cases}
\end{equation}

Let us denote the frequency domain channel response
\begin{align}
H[c] = \sum\limits_{p = 0}^{P - 1} {{h}[p]{e^{ - j\frac{{2\pi cp}}{MN}}}} ,\;c = 0, \cdots ,MN - 1,
\end{align}
and define the diagonal matrix ${{{\bf{\bar H}}}_k} \in {\mathbb{C}^{M \times M}}$ via
\begin{align}\label{H_diag}
{{{\bf{\bar H}}}_k} = \text{diag}\left\{ {H[k],H[k + N], \cdots ,H[k + (M - 1)N]} \right\}.
\end{align}
The following proposition summarizes the findings:
\begin{proposition}\label{OTFS_VOFDM}
For stationary user, the OTFS input-output relationship in (\ref{relate_Static}) is equivalent to
\begin{align}\label{relate_Static_equvalient}
{{\bf{y}}_k} = {{\bf{H}}_k}{{\bf{x}}_k} + {{\bm{\omega }}_k},\;k = 0, \cdots ,N - 1,
\end{align}
where ${{\bf{x}}_k}\in {\mathbb{C}^{M \times 1}}$ and ${{\bf{y}}_k}\in {\mathbb{C}^{M \times 1}}$ are the $k$-th column of ${\bf{X}}$ and ${\bf{Y}}$, respectively. ${{\bm{\omega }}_k} \in {\mathbb{C}^{M \times 1}}$ is the $k$-th noise vector and the equivalent channel matrix ${{\bf{H}}_k}\in {\mathbb{C}^{M \times M}}$ is given by
\begin{align}\label{Static_H}
{{\bf{H}}_k} = {\bf{U}}_k^H{{{\bf{\bar H}}}_k}{{\bf{U}}_k},
\end{align}
where ${{\bf{U}}_k} = {{\bf{F}}_M}{{\bf{\Lambda }}_k}$ is a unitary matrix with
\begin{align}
{{\bf{\Lambda }}_k} = \text{diag}\left\{ {1,{e^{ - j\frac{{2\pi k}}{MN}}},{e^{ - j\frac{{2\pi 2k}}{MN}}}, \cdots ,{e^{ - j\frac{{2\pi (M - 1)k}}{MN}}}} \right\}.
\end{align}
\end{proposition}
\begin{proof}
See Appendix.
\end{proof}

The input-output relationship in (\ref{relate_Static_equvalient}) can be rewritten as
\begin{align}\label{input_output_S}
{{{\bf{\tilde y}}}_{\cal S}} = {{{\bf{\tilde H}}}_{\cal S}}{{{\bf{\tilde x}}}_{\cal S}} + {{\bm{\omega }}_\mathcal{S}},
\end{align}
where ${{{\bf{\tilde x}}}_{\cal S}} = {\left[ {{\bf{x}}_0^T,{\bf{x}}_1^T, \cdots ,{\bf{x}}_{N - 1}^T} \right]^T}$, ${{{\bf{\tilde y}}}_{\cal S}} = {\left[ {{\bf{y}}_0^T,{\bf{y}}_1^T, \cdots ,{\bf{y}}_{N - 1}^T} \right]^T}$, ${{\bm{\omega }}_\mathcal{S}} = {\left[ {{\bm{\omega }}_0^T,{\bm{\omega }}_1^T, \cdots ,{\bm{\omega }}_{N - 1}^T} \right]^T}$ and ${{{\bf{\tilde H}}}_{\cal S}} = \text{diag}\left\{ {{{\bf{H}}_0},{{\bf{H}}_1}, \cdots ,{{\bf{H}}_{N - 1}}} \right\}$.

\vspace{-0.3cm}
\begin{remark}
Note that the mathematical description of OTFS transmission with rectangular pulses for stationary user in \textbf{Proposition \ref{OTFS_VOFDM}} is equivalent to the Vector OFDM \cite{zhang2007asymmetric,xia2001precoded,li2012performance}, which is proposed to as a bridge between conventional OFDM and single carrier modulations. We also notice that similar result has been mentioned in \cite{raviteja2019otfs} according to the equivalent system structures of OTFS and Vector OFDM for stationary user. Such intrinsic equivalence between OTFS and Vector OFDM offers a
new insight for better understanding of Vector OFDM. They share the same characteristics and properties directly.
\end{remark}

\vspace{-0.7cm}
\begin{remark}\label{Remark_property}
From OTFS signal models, we observe that the received signal is only affected by a quasi-stationary channel in the delay-Doppler domain as shown in (\ref{relate_Mobility}) for the mobile users. For the stationary users, OTFS converts an inter-symbol interference (ISI) channel into multiple ``ISI-free'' vector channels as in (\ref{relate_Static_equvalient}). Evidently, OTFS helps simplify the system models for both the
mobile and stationary users.
\end{remark}

\vspace{-0.9cm}

\section{OBNOMA System Model and Resource Allocation}\label{III_Model}

\subsection{Mobility-Profile based User Grouping in OBNOMA}
\begin{figure}%[bth]
  \centering
  \includegraphics[width=6.5in]{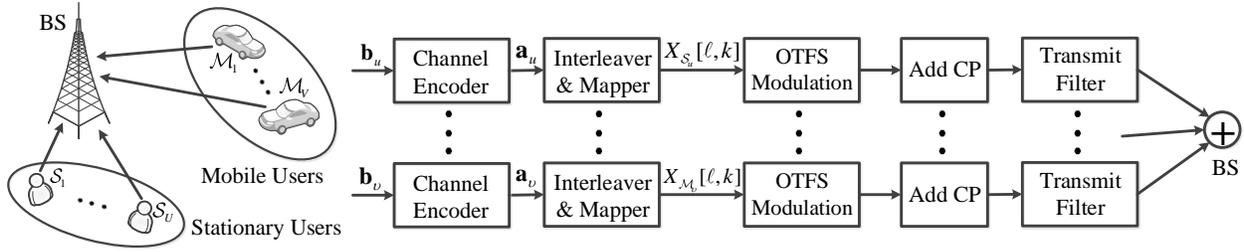}
  \caption{System model for coded uplink OBNOMA scheme.}\label{OTFS_system}
\end{figure}
Consider a coded uplink multiuser system with $(U+V)$ users communicating with a base station (BS) simultaneously as shown in Fig. \ref{OTFS_system}, where $\mathcal{U} = \left\{ {1,2, \cdots ,U} \right\}$ denotes the set of stationary users and $\mathcal{V} = \left\{ {U + 1,U + 2, \cdots ,U + V} \right\}$ represents the set of mobile users. For simplicity, we use $u = \left\{ {1,2, \cdots ,U} \right\}\in \mathcal{U} $ and $v = \left\{ {1,2, \cdots , V} \right\}\in \mathcal{V} $ to denote the $u$-th stationary user and $v$-th mobile user, respectively. To avoid unnecessary confusion, we use a simple model 
in which the user terminals and the BS receiver are equipped with a single 
transmit antenna and receive antenna. Each user utilizes OTFS for uplink transmission to take advantage of its benefits (\textbf{Remark \ref{Remark_property}}). 
Naturally, our model also applies to the cases involving multiple transmit
and receive antennas, with expected diversity gain. 

In practice, the stationary and mobile users experience different
Doppler shifts and fading rates. This channel difference 
allows us to develop a special OTFS-based NOMA (OBNOMA)
by grouping users with different mobility profiles.
Note that laissez faire resource allocation may lead to 
severe CCI. Considering characteristics of stationary and mobile users in terms of Doppler shifts and channel
delay spreads,  we propose a novel resource allocation in OBNOMA scheme 
as shown in Fig. \ref{resource_allocate}. Specifically, the non-overlapping bins along the Doppler axis are assigned to stationary users orthogonally without
mutual interference among these stationary users. On the other hand,
disjoint and contiguous bins along the delay axis 
are allocated to mobile users to mitigate their mutual interference.
\begin{figure}%[bth]
  \centering
  \includegraphics[width=4in]{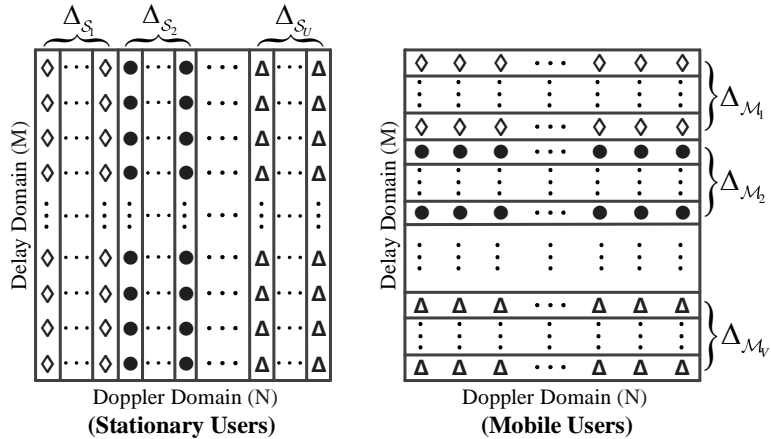}
  \caption{Resource allocation in OBNOMA.}\label{resource_allocate}
\end{figure}

\subsection{Proposed OBNOMA Signal Models}
Consider $K_g$ binary information bits ${{\bf{b}}_g}$ for user $g \in \left\{ {\mathcal{U},\mathcal{V}} \right\}$ which are encoded into a codeword ${{\bf{a}}_g}$ of length $N_g$, where the code rate equals to $K_g/N_g$. The codeword ${{\bf{a}}_g}$ is interleaved to give a data block ${{\bf{d}}_g}$ before being modulated into Gray-mapped symbols ${{\bf{e}}_g}$ drawn from
a complex alphabet $\mathbb{A}$. Thereby, the resulting transmit symbols for the $u$-th stationary user and $v$-th mobile user are placed in the delay-Doppler plane $\Gamma$, and denoted respectively as
\begin{subequations}
\begin{align}\label{symbols_static}
{X_{{\cal S}_u}}[\ell,k] =
\begin{cases}
{{\bf{e}}_u}[\ell + {k_u}N],&\ell \in \{ 0,1, \cdots ,M - 1\} \;\&\; k \in {\Delta _{{\mathcal{S}_u}}},{k_u} = \left\{ {0,1, \cdots ,\left| {{\Delta _{{{\cal S}_u}}}} \right| - 1} \right\}\\
0, &\text{otherwise},
\end{cases}
\end{align}
\begin{align}\label{symbols_mobility}
{X_{{\cal M}_v}}[\ell,k] =
\begin{cases}
{{\bf{e}}_v}[{\ell_v}M + k],&\ell \in {\Delta _{{\mathcal{M}_v}}} \;\&\; k \in \{ 0,1, \cdots ,N - 1\},{\ell_v} = \left\{ {0,1, \cdots ,\left| {{\Delta _{{{\cal M}_v}}}} \right| - 1} \right\}\\
0, &\text{otherwise},
\end{cases}
\end{align}
\end{subequations}
where ${\Delta _{{\mathcal{S}_u}}}$ is the set of bins along the Doppler axis assigned to the $u$-th stationary user of cardinality $\left| {{\Delta _{{{\cal S}_u}}}} \right|$ and ${\Delta _{{\mathcal{M}_v}}}$ is the set of bins along the delay axis assigned to the $v$-th mobile user of cardinality $\left| {{\Delta _{{{\cal M}_v}}}} \right|$, respectively.
Without loss of generality, the mean symbol energies are denoted as ${P_{\cal S}}$ and ${P_{\cal M}}$ for stationary and mobile users, respectively, with the corresponding bit energies ${E_{\cal S}}$ and ${E_{\cal M}}$. In addition, the sets assigned to different users with the same mobility profile are mutually independent, i.e., $ \cup _{u = 1}^U{\Delta _{{\mathcal{S}_u}}} = \{ 0,1, \cdots ,N - 1\} $, $ \cup _{v = 1}^V{\Delta _{{\mathcal{M}_v}}} = \{ 0,1, \cdots ,M - 1\} $, ${\Delta _{{\mathcal{S}_u}}} \cap {\Delta _{{\mathcal{S}_{u'}}}} = \emptyset $ and ${\Delta _{{\mathcal{M}_v}}} \cap {\Delta _{{\mathcal{M}_{v'}}}} = \emptyset$ when $u \ne u'$ and $v \ne v'$, respectively.

Each user employs  OTFS modulation (i.e., ISFFT and Heisenberg transform) 
and adds a CP in front of the generated time domain signal. 
After passing through the transmit filter,  each stationary
user signal is sent out over 
channel response
\begin{align}
{h_{{\mathcal{S}_u}}}\left[ p \right] = \sum\limits_{i = 1}^{L_{{\mathcal{S}_u}}} {{h_{{{\mathcal{S}_u}},i}}{{\mathop{\rm P}\nolimits} _\text{{rc}}}(p{T_s} -{t_{{\mathcal{S}_u}}} - {\tau _{{{\mathcal{S}_u}},i}})}, \;  p = 0, \cdots ,P_{\mathcal{S}_u} - 1,
\end{align}
whereas each mobile user signal is transmitted over channel response
\begin{align}
{h_{{\mathcal{M}_v}}}\left[ {c,p} \right] = \sum\limits_{i = 1}^{L_{{\mathcal{M}_v}}} {{h_{{{\mathcal{M}_v}},i}}{e^{j2\pi {\nu _{{{\mathcal{M}_v}},i}}\left( {c{T_s} - p{T_s}} \right)}}{{\mathop{\rm P}\nolimits} _\text{{rc}}}(p{T_s} - {t_{{\mathcal{M}_v}}}-{\tau _{{{\mathcal{M}_v}},i}})}, \;  \begin{array}{l}
 c = 0, \cdots, MN - 1,
\\p = 0, \cdots ,P_{\mathcal{M}_v} - 1.
\end{array}
\end{align}
Recall that $L_{{\mathcal{S}_u}}$ and ${t_{{\mathcal{S}_u}}}$ are the number of multipaths and the amount of timing offset experienced by the $u$-th stationary user;
$h_{{{\mathcal{S}_u}},i}$ and $\tau _{{{\mathcal{S}_u}},i}$ represent
the gain and delay associated with the $i$-th path of the $u$-th stationary user's channel. Similarly, $L_{{\mathcal{M}_v}}$ and ${t_{{\mathcal{M}_v}}}$ denote the number of multipaths and the timing offset experienced by the $v$-th mobile user; $h_{{{\mathcal{M}_v}},i}$, $\tau _{{{\mathcal{M}_v}},i}$ and $\nu _{{{\mathcal{M}_v}},i}$ stand for the complex gain, delay and Doppler frequency shift associated with the $i$-th path of the $v$-th mobile user's channel, respectively. $P_{\mathcal{S}_u}$ and $P_{\mathcal{M}_v}$ represent, respectively, the channel taps of the $u$-th stationary user and $v$-th mobile user.

We assume that the CP is sufficiently long to accommodate both the maximum timing offset and the maximal channel delay spread. Hence, there is no inter-frame interference. At the receiver, the CP is removed after the received filter. 
We apply the standard Wigner transform and SFFT structure to demodulate OTFS
signals in the delay-Doppler domain, where the input-output relationship can be expressed as
\begin{subequations}\label{input_output_overall}
\begin{align}
{\bf{y}} &= \sum\limits_{u = 1}^U {{{{\bf{\tilde H}}}_{{\mathcal{S}_u}}}{{{\bf{\tilde x}}}_{{\mathcal{S}_u}}}}  + \sum\limits_{v = 1}^V {{{{\bf{\tilde H}}}_{{\mathcal{M}_v}}}{{{\bf{\tilde x}}}_{{\mathcal{M}_v}}}}  + {\bm{\omega }}\\
&= {{{\bf{\bar H}}}_\mathcal{S}}{{{\bf{\bar x}}}_\mathcal{S}} + {{{\bf{\bar H}}}_\mathcal{M}}{{{\bf{\bar x}}}_\mathcal{M}} + {\bm{\omega }},
\end{align}
\end{subequations}
Here, we have used the following notations 
\begin{eqnarray*} &{{{\bf{\bar H}}}_\mathcal{S}} = \left[ {{{{\bf{\tilde H}}}_{{\mathcal{S}_1}}},{{{\bf{\tilde H}}}_{{\mathcal{S}_2}}}, \cdots ,{{{\bf{\tilde H}}}_{{\mathcal{S}_U}}}} \right] \in {\mathbb{C}^{MN \times UMN}},
&  {{{\bf{\bar x}}}_\mathcal{S}} = {\left[ {{\bf{\tilde x}}_{{\mathcal{S}_1}}^T,{\bf{\tilde x}}_{{\mathcal{S}_2}}^T, \cdots ,{\bf{\tilde x}}_{{\mathcal{S}_U}}^T} \right]^T} \in {\mathbb{C}^{UMN \times 1}},
\\
&{{{\bf{\bar H}}}_\mathcal{M}}  = \left[ {{{{\bf{\tilde H}}}_{{\mathcal{M}_1}}},{{{\bf{\tilde H}}}_{{\mathcal{M}_2}}}, \cdots ,{{{\bf{\tilde H}}}_{{\mathcal{M}_V}}}} \right] \in {\mathbb{C}^{MN \times VMN}},
& {{{\bf{\bar x}}}_\mathcal{M}} = {\left[ {{\bf{\tilde x}}_{{\mathcal{M}_1}}^T,{\bf{\tilde x}}_{{\mathcal{M}_2}}^T, \cdots ,{\bf{\tilde x}}_{{\mathcal{M}_V}}^T} \right]^T} \in {\mathbb{C}^{VMN \times 1}}.
\end{eqnarray*}
Also, ${\bf{y}} \in {\mathbb{C}^{MN \times 1}}$ represents the received signal and ${\bm{\omega }} \in {\mathbb{C}^{MN \times 1}} \sim \mathcal{CN}\left( {{\bf{0}},{{\bm{\Sigma }}_{\bm{\omega }}}} \right)$ denotes the noise vector. ${{{\bf{\tilde x}}}_{{\mathcal{S}_u}}} \in {\mathbb{C}^{MN \times 1}}$ and ${{{\bf{\tilde x}}}_{{\mathcal{M}_v}}} \in {\mathbb{C}^{MN \times 1}}$ contain the transmitted symbols from $u$-th stationary user and $v$-th mobile user, respectively. The equivalent channels ${{{{\bf{\tilde H}}}_{{\mathcal{S}_u}}}}$ and ${{{{\bf{\tilde H}}}_{{\mathcal{M}_v}}}}$ have the similar structures as ${{{\bf{\tilde H}}}_{\cal S}}$ in (\ref{input_output_S}) and ${{{\bf{\tilde H}}}_{\cal M}}$ in (\ref{input_output_M}), respectively.

Note that ${{{\bf{\bar x}}}_\mathcal{S}}$ and ${{{\bf{\bar x}}}_\mathcal{M}}$ are sparse vectors due to the resource allocations of (\ref{symbols_static}) and (\ref{symbols_mobility}). The numbers of non-zero elements in ${{{\bf{\bar x}}}_\mathcal{S}}$ and ${{{\bf{\bar x}}}_\mathcal{M}}$ are identically $MN$. Let ${{\bf{x}}_\mathcal{S}} \in {\mathbb{C}^{MN \times 1}}$ and ${{\bf{x}}_\mathcal{M}} \in {\mathbb{C}^{MN \times 1}}$ denote the effective input vectors after removing the zeros in ${{{\bf{\bar x}}}_\mathcal{S}}$ and ${{{\bf{\bar x}}}_\mathcal{M}}$; Let ${{\bf{H}}_\mathcal{S}} \in {\mathbb{C}^{MN \times MN}}$ and ${{\bf{H}}_\mathcal{M}} \in {\mathbb{C}^{MN \times MN}}$ represent the effective matrices after deleting the columns corresponding to the indices of the zeros in ${{{\bf{\bar x}}}_\mathcal{S}}$ and ${{{\bf{\bar x}}}_\mathcal{M}}$, respectively. We can then simplify the relationship of (\ref{input_output_overall}) to
\begin{align}\label{input_output_simplify}
{\bf{y}} = {{{\bf{ H}}}_\mathcal{S}}{{{\bf{ x}}}_\mathcal{S}} + {{{\bf{ H}}}_\mathcal{M}}{{{\bf{ x}}}_\mathcal{M}} + {\bm{\omega }}.
\end{align}

From (\ref{input_output_simplify}), we observe that the conventional single user detection \cite{tiwari2019low,raviteja2018interference,yuan2020simple} or multi-user detection with OMA \cite{augustine2019interleaved,surabhi2019multiple} cannot be directly applied to recover the signal due to the strong presence of CCI at the receiver. Therefore, additional processing such as iterative SIC techniques should be employed to mitigate CCI effect.

\section{Iterative SIC Turbo Receiver for OBNOMA}\label{IV_Receiver}
We now investigate the recovery of the signal for each user from the received aggregated signal at the BS. Here, we propose an iterative SIC turbo receiver to overcome the CCI and self-interference in delay-Doppler domain.

\subsection{Receiver Structure}
\begin{figure}%[bth]
  \centering
  \includegraphics[width=6.2in]{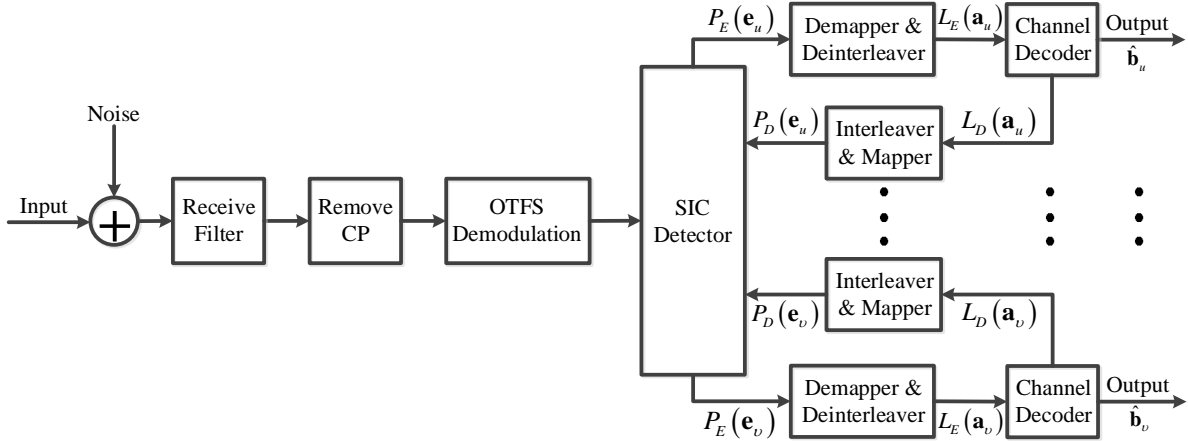}
  \caption{Iterative SIC turbo receiver structure.}\label{receiver_structure}
\end{figure}
The structure of iterative SIC turbo receiver is shown in Fig. \ref{receiver_structure}. The key point is the iterative exchange of information between the SIC detector and the co-channel individual user channel decoders. The extrinsic information generated from the SIC detector is treated as {\em a priori} information by the individual channel decoders. Next, the channel decoders generate extrinsic information to be used by the SIC detector as its {\em a priori} information to form a soft-input-soft-output turbo processing loop. 
Note that the concept of turbo equalization has been extensively studied for stationary communication systems \cite{tuchler2011turbo} such as single-user \cite{douillard1995} or multi-user scenarios \cite{wang1999iterative,li2005exit}, 
which can achieve excellent performance.

Without loss of generality, we
let $\mathcal{R} = {\log _2}Q$ be the number of bits in each symbol.
In this receiver, the SIC detector generates extrinsic probabilities ${P_E}\left( {{{\bf{e}}_g}} \right),\forall g \in \left\{ {{\cal U},{\cal V}} \right\}$ by taking the feedback information ${P_D}\left( {{{\bf{e}}_g}} \right)$ from the channel decoders as the input {\em a priori} probabilities.
To initialize, the SIC detector starts with equiprobable symbols
without prior information from channel decoders. 
These extrinsic probabilities ${P_E}\left( {{{\bf{e}}_g}} \right)$ are then demapped and their bit log-likelihood ratios (LLRs) can be expressed as
\begin{align}\label{demap_LLR}
{L_E}\left( {{{\bf{d}}_g}[c\mathcal{R} + j]} \right) = \log \frac{{\sum\nolimits_{\left. {{{\bf{e}}_g}[c] \in \mathbb{A}} \right|{{\bf{d}}_g}[c\mathcal{R} + j] = 0} {{P_E}\left( {{{\bf{e}}_g}[c]} \right)} }}{{\sum\nolimits_{\left. {{{\bf{e}}_g}[c] \in \mathbb{A}} \right|{{\bf{d}}_g}[c\mathcal{R} + j] = 1} {{P_E}\left( {{{\bf{e}}_g}[c]} \right)} }}, \; c = 0, \cdots ,\frac{{{N_g}}}{\mathcal{R}} - 1,
\end{align}
where ${{{\bf{d}}_g}[c\mathcal{R} + j]}$ denoting the $j$-th bit associated to the $c$-th symbol of the $g$-th user. 
The resulted LLRs are deinterleaved as ${L_E}\left( {{{\bf{a}}_g}} \right)$ before passing to the channel decoders. The channel decoder computes an estimation of the information bits ${{{\bf{\hat b}}}_g}$, along with the extrinsic LLRs on the coded bits ${{{\bf{a}}_g}}$ according to
\begin{align}\label{chan_code}
{L_D}\left( {{{\bf{a}}_g}[i]} \right) = \log \frac{{\Pr \left( {\left. {{{\bf{a}}_g}[i] = 0} \right|{L_E}\left( {{{\bf{a}}_g}} \right)} \right)}}{{\Pr \left( {\left. {{{\bf{a}}_g}[i] = 1} \right|{L_E}\left( {{{\bf{a}}_g}} \right)} \right)}} - {L_E}\left( {{{\bf{a}}_g}[i]} \right), \; i = 1,2, \cdots ,{N_g}.
\end{align}
These extrinsic LLRs are interleaved as ${L_D}\left( {{{\bf{d}}_g}} \right)$ and mapped back to update the input {\em a priori} probabilities of SIC detector
\begin{align}\label{map_LLR}
{P_D}\left( {{{\bf{e}}_g}[c] = \chi } \right) \propto \prod _{j = 1}^\mathcal{R}{e^{ - \varphi _j^{ - 1}(\chi ){L_D}\left( {{{\bf{d}}_g}[c\mathcal{R} + j]} \right)}},
\end{align}
where ${\varphi _j^{ - 1}(\chi )}$ denotes the value of the $j$-th bit labelling the symbol $\chi  \in \mathbb{A}$.

This ``turbo'' message passing process is repeated iteratively before terminating
at a maximum iteration number $n_t$ or upon meeting other preset stopping criteria.
A detailed implementation of the iterative SIC turbo receiver is summarized in \textbf{Algorithm \ref{alg:A}}.
Here the SIC detector is adopted, where we detect stationary users' signal first
before removing their contribution to the aggregated received signal via SIC.
The mobile users' signal are then detected upon the removal of stationary users' signal. 
Based on different system models, we propose two detection algorithms next.
\begin{algorithm}%[H]
\caption{Iterative SIC Turbo Receiver}
\label{alg:A}
\begin{algorithmic}
\STATE {Initialization: ${P_D}\left( {{{\bf{e}}_g}[c] = \chi } \right) = {1 \mathord{\left/
 {\vphantom {1 Q}} \right.
 \kern-\nulldelimiterspace} Q},c = 0, \cdots ,\frac{{{N_g}}}{\mathcal{R}} - 1,\forall g \in \left\{ {{\cal U},{\cal V}} \right\}, \chi  \in \mathbb{A}$.}
\FOR{$\mathcal{T} = 0,1, \cdots ,n_t$}
\STATE {\textbf{SIC detector:}}
\STATE {1)\; Obtain $P\left( {{{\bf{x}}_{{\mathcal{S}}}}} \right)$ and ${P_E}\left( {{{\bf{e}}_u}} \right),\forall u \in {\cal U}$ by employing OAMP-LMMSE detector in \textbf{Algorithm \ref{alg:B}};}
\STATE {2)\; Obtain ${P_E}\left( {{{\bf{e}}_v}} \right),\forall v \in {\cal V}$ by employing GAMP-EP detector in \textbf{Algorithm \ref{alg:C}};}
\STATE {3)\; Demap the extrinsic probabilities ${P_E}\left( {{{\bf{e}}_g}} \right),\forall g \in \left\{ {{\cal U},{\cal V}} \right\}$ and compute their extrinsic bit LLRs ${L_E}\left( {{{\bf{d}}_g}} \right),\forall g \in \left\{ {{\cal U},{\cal V}} \right\}$ in (\ref{demap_LLR});}
\STATE {4)\; Deinterleave ${L_E}\left( {{{\bf{d}}_g}} \right),\forall g \in \left\{ {{\cal U},{\cal V}} \right\}$ and deliver them to the channel decoders;}
\STATE {\textbf{Channel decoders:}}
\STATE {5)\; Run each channel decoder to output ${L_D}\left( {{{\bf{a}}_g}} \right),\forall g \in \left\{ {{\cal U},{\cal V}} \right\}$ in (\ref{chan_code});}
\STATE {6)\; Interleave the extrinsic LLRs ${L_D}\left( {{{\bf{a}}_g}} \right),\forall g \in \left\{ {{\cal U},{\cal V}} \right\}$ and map them to obtain ${P_D}\left( {{{\bf{e}}_g}} \right),\forall g \in \left\{ {{\cal U},{\cal V}} \right\}$ in (\ref{map_LLR});}
\ENDFOR
\STATE {Output: The decisions of the information bits ${{{\bf{\hat b}}}_g},\forall g \in \left\{ {{\cal U},{\cal V}} \right\}$ from the channel decoders.}
\end{algorithmic}
\end{algorithm}

\subsection{OAMP-LMMSE Detector for Stationary Users in OBNOMA}
After obtaining the feedback information ${P_D}\left( {{{\bf{e}}_g}} \right),\forall g \in \left\{ {{\cal U},{\cal V}} \right\}$ from the co-channel individual user channel decoders, we first initialize the {\em a priori} probabilities of ${P_D}\left( {{{\bf{x}}_{\cal S}}} \right)$ and ${P_D}\left( {{{\bf{x}}_{\cal M}}} \right)$ according to the indices of symbols for each user corresponding to ${P_D}\left( {{{\bf{e}}_g}} \right),\forall g \in \left\{ {{\cal U},{\cal V}} \right\}$. We then project each entry of probabilities ${P_D}\left( {{{\bf{x}}_{\cal M}}} \right)$ into Gaussian distribution, with respective mean and variance given by
\begin{align}\label{Gaussi_pro}
{\mu _{{\mathcal{M}_c}}}  = \sum\limits_{\,\chi  \in \mathbb{A}} {\chi {P_D}\left( {{x_{{\mathcal{M}_c}}} = \chi } \right)}, \quad
{\eta _{{\mathcal{M}_c}}} = \sum\limits_{\,\chi  \in \mathbb{A}} {{{\left| \chi  \right|}^2}{P_D}\left( {{x_{{\mathcal{M}_c}}} = \chi } \right)}  - {\left| {{\mu _{{\mathcal{M}_c}}}} \right|^2},
\end{align}
for $c = 0,1, \cdots ,MN - 1$.
Next, we can approximately rewrite (\ref{input_output_simplify}) as
\begin{align}\label{static_form}
{{\bf{y}}_\mathcal{S}}  \simeq  {{\bf{H}}_{\cal S}}{{\bf{x}}_{\cal S}} + {{\bf{z}}_\mathcal{S}},
\end{align}
where ${{\bf{y}}_\mathcal{S}} = {\bf{y}} - {{\bf{H}}_{\cal M}}{{\bm{\mu }}_\mathcal{M}}$ and ${{\bf{z}}_\mathcal{S}}$ is modeled as $\mathcal{CN}\left( {{\bf{0}},{{\bf{\Sigma }}_\mathcal{S}}} \right)$ with covariance matrix ${{\bm{\Sigma }}_\mathcal{S}} = {{\bf{\Sigma }}_{\bm{\omega }}} + {{\bf{H}}_{\cal M}}\text{diag}\left\{ {{{\bm{\eta }}_\mathcal{M}}} \right\}{\bf{H}}_\mathcal{M}^H$.

For $m = 0,\;1,\; \cdots,\;M - 1$, we define the following notations 
\begin{eqnarray*} &{\bf{x}}_{{\cal S}_k} = {\bf{x}}_{{\cal S}}\left[ {kN + m} \right],
\quad  {{\bf{y}}_{{\mathcal{S}_k}}} = {\bf{y}}_{\mathcal{S}}\left[ {kN + m} \right], \quad {{\bf{z}}_{{\mathcal{S}_k}}}  \sim \mathcal{CN}
\left({\bf{0}},{\bf{\Sigma }}_{\mathcal{S}_k}\right), 
\\
& {{\bf{H}}_{{\mathcal{S}_k}}} = {\bf{H}}_{\mathcal{S}}\left[ {kN + m,kN + m} \right],
\quad  {\bf{\Sigma }}_{\mathcal{S}_k} = {\bf{\Sigma }}_{\mathcal{S}}\left[ {kN + m,kN + m} \right].
\end{eqnarray*}
As a result, we can rewrite (\ref{static_form})
as multiple ``ISI-free'' vector channels as in (\ref{relate_Static_equvalient}),
\begin{align}\label{static_form_block}
{{\bf{y}}_{{\mathcal{S}_k}}} \simeq {{\bf{H}}_{{\mathcal{S}_k}}}{{\bf{x}}_{{\mathcal{S}_k}}} + {{\bf{z}}_{{\mathcal{S}_k}}},\;k = 0,1, \cdots ,N - 1.
\end{align}

Through this process, the stationary users' signal can be detected block by block as in (\ref{static_form_block}). One can apply conventional linear receivers such as zero-forcing (ZF) and LMMSE \cite{surabhi2019low,tiwari2019low,li2012performance} or an efficient message passing (MP) algorithm \cite{raviteja2018interference,raviteja2019otfs} for symbol detection.

Recently, an OAMP algorithm is proposed in \cite{ma2019orthogonal}, where the extrinsic messages passed iteratively in the factor graph are only required to be orthogonal rather than stringent independent in the original MP.  
The successful performance improvement of OAMP motivates us to detect stationary users' signal by combining OAMP with LMMSE.

%\begin{figure}%[bth]
%  \centering
%  \includegraphics[width=2.8in]{FG_S.pdf}
%  \caption{Factor graph describing (\ref{post_static}).}\label{FGS}
%\end{figure}

To describe the detail steps of the OAMP-LMMSE receiver, 
we focus on the $k$-th block without loss of generality.
Specifically, we use ${{{x}}_{{\mathcal{S}_{k,m}}}}$ to denote the $m$-th symbol in the $k$-th block and assume its {\em a priori} distribution to be Gaussian, modeled as ${q_D}\left( {{{{x}}_{{\mathcal{S}_{k,m}}}}} \right) \sim \mathcal{CN}\left( {{\mu _{{\mathcal{S}_{k,m}}}},{\eta _{{\mathcal{S}_{k,m}}}}} \right)$. Hence, the joint posteriori distribution can be decomposed as follows:
\begin{align}\label{post_static}
p\left( {{{\bf{x}}_{{\mathcal{S}_k}}}\left| {{{\bf{y}}_{{\mathcal{S}_k}}},{{\bf{H}}_{{\mathcal{S}_k}}}} \right.} \right) \propto p\left( {{{\bf{y}}_{{\mathcal{S}_k}}}\left| {{{\bf{x}}_{{\mathcal{S}_k}}},{{\bf{H}}_{{\mathcal{S}_k}}}} \right.} \right)\prod\limits_{m = 0}^{M - 1} {{q_D}\left( {{{{x}}_{{\mathcal{S}_{k,m}}}}} \right)}  \sim \mathcal{CN}\left( {{{\bf{A}}_{{\mathcal{S}_k}}},{{\bf{B}}_{{\mathcal{S}_k}}}} \right).
\end{align}
Here, we can represent (\ref{post_static}) by using a factor graph, where a factor node ${{\bf{y}}_{{\mathcal{S}_k}}}$ is connected to multiple variable nodes ${x_{{\mathcal{S}_{k,m}}}},m = 0,1, \cdots ,M - 1$. We approximate the {\em a posteriori} distribution by computing and passing messages between the factor node ${{{\bf{y}}_{{\mathcal{S}_k}}}}$ and variable nodes ${x_{{\mathcal{S}_{k,m}}}},m = 0,1, \cdots ,M - 1$ in this factor graph iteratively. The OAMP-LMMSE detector is summarized in \textbf{Algorithm \ref{alg:B}}. We now describe its detailed steps in iteration $\iota $:
\begin{algorithm}%[H]
\caption{OAMP-LMMSE Detector}
\label{alg:B}
\begin{algorithmic}
\STATE {Input: ${P_D}\left( {{{\bf{e}}_g}} \right),\forall g \in \left\{ {{\cal U},{\cal V}} \right\}$, $\bf{y}$, ${{{\bf{ H}}}_\mathcal{S}}$, ${{{\bf{ H}}}_\mathcal{M}}$ and ${n_\mathcal{S}}$.}
\STATE {Initialization: ${P_D}\left( {{{\bf{x}}_{\cal S}}} \right)$, ${P_D}\left( {{{\bf{x}}_{\cal M}}} \right)$, ${{\bm{\mu }}_\mathcal{M}}$, ${{\bm{\eta }}_\mathcal{M}}$, ${{\bf{y}}_\mathcal{S}}$, ${{\bm{\mu }}_{{\mathcal{S}}}^{(0)}}$, ${{\bm{\eta }}_{{\mathcal{S}}}^{(0)}}$ and $\alpha _\mathcal{S}^{(0)}=0$.}
\FOR{$k = 0,1, \cdots ,N - 1$}
\STATE {Set iteration count $\iota=1$.}
\REPEAT
\STATE {1)\; Factor node ${{\bf{y}}_{{\mathcal{S}_k}}}$ generates the extrinsic mean $C_{{\mathcal{S}_{k,m}}}^{(\iota )}$ and variance $D_{{\mathcal{S}_{k,m}}}^{(\iota )}$ in (\ref{extrinsic_ys}), then sends them to the variable nodes ${x_{{\mathcal{S}_{k,m}}}},m = 0,1, \cdots ,M - 1$;}
\STATE {2)\; Each variable node ${x_{{\mathcal{S}_{k,m}}}}$ computes the mean $\mu _{{\mathcal{S}_{k,m}}}^{(\iota )}$ and variance $\eta_{{\mathcal{S}_{k,m}}}^{(\iota )}$ in (\ref{extrinsic_xs}), and passes them back to the factor node ${{\bf{y}}_{{\mathcal{S}_k}}}$;}
\STATE {3)\; Calculate the convergence indicator $\alpha _\mathcal{S}^{(\iota )}$ in (\ref{conv_s});}
\STATE {4)\; Update $P\left( {{{\bf{x}}_{{\mathcal{S}_k}}}} \right) = {{\bar P}^{(\iota )}}\left( {{{\bf{x}}_{{\mathcal{S}_k}}}} \right)$ and ${P_E}\left( {{{\bf{x}}_{{\mathcal{S}_k}}}} \right) = \bar P_E^{(\iota )}\left( {{{\bf{x}}_{{\mathcal{S}_k}}}} \right)$ if $\alpha _\mathcal{S}^{(\iota )}>\alpha _\mathcal{S}^{(\iota -1)}$;}
\STATE {5)\; $\iota: = \iota + 1$;}
\UNTIL{$\alpha _\mathcal{S}^{(\iota )}=1$ or $\iota={n_\mathcal{S}}$.}
\ENDFOR
\STATE {Output: $P\left( {{{\bf{x}}_{{\mathcal{S}}}}} \right)$ and ${P_E}\left( {{{\bf{e}}_u}} \right),\forall u \in {\cal U}$.}
\end{algorithmic}
\end{algorithm}

\textbf{1) From factor node ${{\bf{y}}_{{\mathcal{S}_k}}}$ to variable nodes ${x_{{\mathcal{S}_{k,m}}}},m = 0,1, \cdots ,M - 1$:} For simplicity, we can apply the LMMSE criterion at the factor node to compute the {\em a posteriori} distribution \cite{tiwari2019low,li2012performance,kay1993fundamentals}:
\begin{align}\label{S_Matrix_I}
{\bf{B}}_{{\mathcal{S}_k}}^{(\iota )} = {\left( {{\bf{H}}_{{\mathcal{S}_k}}^H{\bf{\Sigma }}_{{\mathcal{S}_k}}^{ - 1}{{\bf{H}}_{{\mathcal{S}_k}}} + \text{diag}{{\left\{ {\bm{\eta }}_{{\mathcal{S}_k}}^{(\iota -1 )}\right\} }^{ - 1}}} \right)^{ - 1}},
\end{align}
\begin{align}
{\bf{A}}_{{\mathcal{S}_k}}^{(\iota )} = {\bf{B}}_{{\mathcal{S}_k}}^{(\iota )}\left( {{\bf{H}}_{{\mathcal{S}_k}}^H{\bf{\Sigma }}_{{\mathcal{S}_k}}^{ - 1}{{\bf{y}}_{{\mathcal{S}_k}}} + \text{diag}{{\left\{ {\bm{\eta }}_{{\mathcal{S}_k}}^{(\iota-1 )}\right\} }^{ - 1}}{\bm{\mu }}_{{\mathcal{S}_k}}^{(\iota -1)}} \right),
\end{align}
where ${{\bm{\mu }}_{{\mathcal{S}_k}}^{(\iota -1)}}$ and ${{\bm{\eta }}_{{\mathcal{S}_k}}^{(\iota-1 )}}$ are the mean and variance vectors for the symbols of $k$-th block, which are acquired from the variable nodes in the $(\iota -1)$-th iteration and can be initialized in the first iteration by projecting the probabilities ${P_D}\left( {{{\bf{x}}_{{\mathcal{S}_k}}}} \right)$ from the channel decoders into Gaussian distributions. Following the Gaussian message combining rule \cite{loeliger2007factor}, we then update the extrinsic marginal distribution $q_E^{(\iota )}\left( {{x_{{\mathcal{S}_{k,m}}}}} \right) \sim \mathcal{CN}\left( {C_{{\mathcal{S}_{k,m}}}^{(\iota )},D_{{\mathcal{S}_{k,m}}}^{(\iota )}} \right)$, with
\begin{align}\label{extrinsic_ys}
D_{{\mathcal{S}_{k,m}}}^{(\iota )} = {\left[ {{{\left( {B_{{\mathcal{S}_{k,m}}}^{(\iota )}} \right)}^{ - 1}} - {{\left( {\eta _{{\mathcal{S}_{k,m}}}^{(\iota -1)}} \right)}^{ - 1}}} \right]^{ - 1}}, \quad
C_{{\mathcal{S}_{k,m}}}^{(\iota )} = D_{{\mathcal{S}_{k,m}}}^{(\iota )}\left[ {\frac{{A_{{\mathcal{S}_{k,m}}}^{(\iota )}}}{{B_{{\mathcal{S}_{k,m}}}^{(\iota )}}} - \frac{{\mu _{{\mathcal{S}_{k,m}}}^{(\iota -1)}}}{{\eta _{{\mathcal{S}_{k,m}}}^{(\iota -1)}}}} \right],
\end{align}
where ${B_{{\mathcal{S}_{k,m}}}^{(\iota )}}$ denotes the $m$-th diagonal element of ${\bf{B}}_{{\mathcal{S}_k}}^{(\iota )}$. At last, the factor node sends the mean $C_{{\mathcal{S}_{k,m}}}^{(\iota )}$ and variance $D_{{\mathcal{S}_{k,m}}}^{(\iota )}$ to the variable nodes.

\textbf{2) From variable nodes ${x_{{\mathcal{S}_{k,m}}}},m = 0,1, \cdots ,M - 1$ to factor node ${{\bf{y}}_{{\mathcal{S}_k}}}$:} The {\em a posteriori} probability can be decomposed as follows at each variable node
\begin{align}
{{\bar P}^{(\iota )}}\left( {{x_{{\mathcal{S}_{k,m}}}} = \chi } \right) \propto {P_D}\left( {{x_{{\mathcal{S}_{k,m}}}} = \chi } \right)\exp \left( { - \frac{{{{\left| {\chi  - C_{{\mathcal{S}_{k,m}}}^{(\iota )}} \right|}^2}}}{{D_{{\mathcal{S}_{k,m}}}^{(\iota )}}}} \right),\;\forall \chi  \in \mathbb{A},
\end{align}
and then projected into a Gaussian distribution $\mathcal{CN}\left( {E_{{\mathcal{S}_{k,m}}}^{(\iota )},F_{{\mathcal{S}_{k,m}}}^{(\iota )}} \right)$. In order to avoid numerical instabilities, we set a minimum allowed variance $\varepsilon $, i.e., $F_{{\mathcal{S}_{k,m}}}^{(\iota )} = \max \{ \varepsilon ,F_{{\mathcal{S}_{k,m}}}^{(\iota )}\} $. Then, we can update the extrinsic distribution $\bar q_E^{(\iota )}\left( {{x_{{\mathcal{S}_{k,m}}}}} \right) \sim \mathcal{CN}\left( {\bar \mu _{{\mathcal{S}_{k,m}}}^{(\iota )},\bar \eta _{{\mathcal{S}_{k,m}}}^{(\iota )}} \right)$, where
\begin{align}
\bar \eta _{{\mathcal{S}_{k,m}}}^{(\iota )} = {\left[ {{{\left( {F_{{\mathcal{S}_{k,m}}}^{(\iota )}} \right)}^{ - 1}} - {{\left( {D_{{\mathcal{S}_{k,m}}}^{(\iota )}} \right)}^{ - 1}}} \right]^{ - 1}}, \quad
\bar \mu _{{\mathcal{S}_{k,m}}}^{(\iota )} = \bar \eta _{{\mathcal{S}_{k,m}}}^{(\iota )}\left[ {\frac{{E_{{\mathcal{S}_{k,m}}}^{(\iota )}}}{{F_{{\mathcal{S}_{k,m}}}^{(\iota )}}} - \frac{{C_{{\mathcal{S}_{k,m}}}^{(\iota )}}}{{D_{{\mathcal{S}_{k,m}}}^{(\iota )}}}} \right].
\end{align}
To improve the performance and control the convergence speed of the algorithm, we apply a damping factor ${\delta _\mathcal{S}} \in \left( {0,1} \right]$ \cite{minka2005divergence,santos2018turbo,csahin2018iterative}, i.e.,
\begin{align}\label{extrinsic_xs}
\eta _{{\mathcal{S}_{k,m}}}^{(\iota )} = {\left[ {\frac{{{\delta _\mathcal{S}}}}{{\bar \eta _{{\mathcal{S}_{k,m}}}^{(\iota )}}} + \frac{{(1 - {\delta _\mathcal{S}})}}{{\eta _{{\mathcal{S}_{k,m}}}^{(\iota -1)}}}} \right]^{ - 1}}, \quad
\mu _{{\mathcal{S}_{k,m}}}^{(\iota )} = \eta _{{\mathcal{S}_{k,m}}}^{(\iota )}\left[ {{\delta _\mathcal{S}}\frac{{\bar \mu _{{\mathcal{S}_{k,m}}}^{(\iota )}}}{{\bar \eta _{{\mathcal{S}_{k,m}}}^{(\iota )}}} + (1 - {\delta _\mathcal{S}})\frac{{\mu _{{\mathcal{S}_{k,m}}}^{(\iota -1)}}}{{\eta _{{\mathcal{S}_{k,m}}}^{(\iota -1)}}}} \right].
\end{align}
If the renewed variance $\eta _{{\mathcal{S}_{k,m}}}^{(\iota )}$ 
becomes negative, we would skip this update. Finally, $\mu _{{\mathcal{S}_{k,m}}}^{(\iota )}$ and $\eta_{{\mathcal{S}_{k,m}}}^{(\iota )}$ are passed back to the factor node.

\textbf{3) Convergence indicator:} We define a convergence indicator $\alpha _\mathcal{S}^{(\iota )}$ for stationary users as  
\begin{align}\label{conv_s}
\alpha _\mathcal{S}^{(\iota )} = \frac{1}{M}\sum\limits_{m = 0}^{M - 1} {\mathbb{I}\left( {\mathop {\max }\limits_{\chi  \in \mathbb{A}} {{\bar P}^{(\iota )}}\left( {{x_{{\mathcal{S}_{k,m}}}} = \chi } \right) \geq 1-\varrho } \right)},
\end{align}
where $\mathbb{I}(\cdot)$ represents the indicator function and $\varrho >0$ is a small value.

\textbf{4) Update criterion:} If $\alpha _\mathcal{S}^{(\iota )} > \alpha _\mathcal{S}^{(\iota  - 1)}$, we update
\begin{align}
P\left( {{{\bf{x}}_{{\mathcal{S}_k}}}} \right) = {{\bar P}^{(\iota )}}\left( {{{\bf{x}}_{{\mathcal{S}_k}}}} \right),\; {P_E}\left( {{{\bf{x}}_{{\mathcal{S}_k}}}} \right) = \bar P_E^{(\iota )}\left( {{{\bf{x}}_{{\mathcal{S}_k}}}} \right),
\end{align}
where $\bar P_E^{(\iota )}\left( {{x_{{\mathcal{S}_{k,m}}}} = \chi } \right) \propto \exp \left( { - \frac{{{{\left| {\chi  - C_{{\mathcal{S}_{k,m}}}^{(\iota )}} \right|}^2}}}{{D_{{\mathcal{S}_{k,m}}}^{(\iota )}}}} \right),\forall \chi  \in \mathbb{A}$.

\textbf{5) Stopping criterion:} The OAMP-LMMSE detector terminates when either $\alpha _\mathcal{S}^{(\iota )} = 1$ or the maximum iteration number ${n_\mathcal{S}}$ is reached.

Note that the OAMP-LMMSE can be used to detect each block's symbols parallelly, thus, the delay to detect the whole stationary users' signal is manageable. Finally, we obtain the extrinsic probabilities of ${P_E}\left( {{{\bf{e}}_u}} \right),\forall u \in {\cal U}$ according to the indices of each stationary user's symbols corresponding to ${P_E}\left( {{{\bf{x}}_{{\mathcal{S}}}}} \right)$, and output $P\left( {{{\bf{x}}_{{\mathcal{S}}}}} \right)$ and ${P_E}\left( {{{\bf{e}}_u}} \right),\forall u \in {\cal U}$.

\subsection{GAMP-EP Detector for Mobile Users in OBNOMA}
With the {\em a posteriori} probabilities $P\left( {{{\bf{x}}_{{\mathcal{S}}}}} \right)$ of the stationary users' symbols from OAMP-LMMSE detector, we first project each entry of these probabilities into Gaussian distribution, denoted as $\hat q\left( {{x_{{\mathcal{S}_c}}}} \right) \sim \mathcal{CN}\left( {{{\hat \mu }_{{\mathcal{S}_c}}},{{\hat \eta }_{{\mathcal{S}_c}}}} \right),c = 0,1, \cdots ,MN - 1$. From (\ref{input_output_simplify}), we can approximate
\begin{align}\label{mobile_form}
{{\bf{y}}_\mathcal{M}}  \simeq  {{\bf{H}}_{\cal M}}{{\bf{x}}_{\cal M}} + {{\bf{z}}_\mathcal{M}},
\end{align}
where ${{\bf{y}}_\mathcal{M}} = {\bf{y}} - {{\bf{H}}_{\cal S}}{{{\bm{\hat \mu }}}_\mathcal{S}}$ and ${{\bf{z}}_\mathcal{M}}$ is modeled as $\mathcal{CN}\left( {{\bf{0}},{{\bf{\Sigma }}_\mathcal{M}}} \right)$ with covariance matrix ${{\bf{\Sigma }}_\mathcal{M}} = {{\bf{\Sigma }}_{\bm{\omega }}} + {{\bf{H}}_\mathcal{S}}\text{diag}\left\{ {{{{\bm{\hat \eta }}}_\mathcal{S}}} \right\}{\bf{H}}_\mathcal{S}^H$.

Direct solution of (\ref{mobile_form}) by employing OAMP-LMMSE detector could be computationally costly since it involves a large matrix inverse while the typical value of $MN$ can be in the order of thousands or even larger in OTFS system. Fortunately, ${{\bf{H}}_{\cal M}}$ is a sparse matrix and the index sets of non-zero components in the $d$-th row and $c$-th column can be denoted as $\mathcal{I}(d)$, $d = 0,1, \cdots ,MN - 1$ and $\mathcal{J}(c)$, $c = 0,1, \cdots ,MN - 1$, respectively. We also represent the corresponding numbers of non-zero components in the $d$-th row and $c$-th column as $\left| {\mathcal{I}(d)} \right|$ and $\left| {\mathcal{J}(c)} \right|$. Hence, we can use a sparsely connected factor graph to describe the system model of (\ref{mobile_form}), where the entries of ${{\bf{y}}_\mathcal{M}}$ and ${{\bf{x}}_{\cal M}}$ are regarded as factor nodes and variable nodes, respectively.

%as shown in Fig. \ref{FGM}.
%\begin{figure}%[bth]
%  \centering
%  \includegraphics[width=3.8in]{FG_M.pdf}
%  \caption{Factor graph describing (\ref{mobile_form}).}\label{FGM}
%\end{figure}

Unlike the existing works in \cite{raviteja2018interference,ramachandran2018mimo,deka2020otfs,ge2020}, 
which use MP for symbol detection. 
Here, we propose a GAMP-EP detector for performance 
improvement. Note that EP algorithm is a Bayesian inference technique 
developed to approximate the true posterior. 
It has been already successfully applied for symbol detection in the stationary communication systems \cite{santos2018turbo,csahin2018iterative} 
with modest complexity.
We approximate the messages updated and passed between the factor nodes and variable nodes on the factor graph as Gaussian. 
\textbf{Algorithm \ref{alg:C}} contains a detailed description 
of GAMP-EP detector, and the steps of the $\kappa $-th iteration are introduced below:
\begin{algorithm}%[H]
\caption{GAMP-EP Detector}
\label{alg:C}
\begin{algorithmic}
\STATE {Input: ${P_D}\left( {{{\bf{e}}_v}} \right),\forall v \in {\cal V}$, $P\left( {{{\bf{x}}_{{\mathcal{S}}}}} \right)$, $\bf{y}$, ${{{\bf{ H}}}_\mathcal{S}}$, ${{{\bf{ H}}}_\mathcal{M}}$ and ${n_\mathcal{M}}$.}
\STATE {Initialization: ${P_D}\left( {{{\bf{x}}_{\cal M}}} \right)$, ${{{\bm{\hat \mu }}}_\mathcal{S}}$, ${{{\bm{\hat \eta }}}_\mathcal{S}}$, ${{\bf{y}}_\mathcal{M}}$, ${\mu _{{\mathcal{M}_{d,c}}}^{(0)}}={\mu _{{\mathcal{M}_c}}}$, ${\eta _{{\mathcal{M}_{d,c}}}^{(0)}}={\eta _{{\mathcal{M}_c}}},c = 0,1, \cdots ,MN - 1,d \in \mathcal{J}(c)$, $\alpha _\mathcal{M}^{(0)}=0$ and iteration count $\kappa=1$.}
\REPEAT
\STATE {1)\; Each factor node ${y_{{\mathcal{M}_d}}}$ generates the mean $C_{{\mathcal{M}_{d,c}}}^{(\kappa )}$ and variance $D_{{\mathcal{M}_{d,c}}}^{(\kappa )}$ in (\ref{extrinsic_ym_m}) and (\ref{extrinsic_ym_v}), then delivers them to the connected variable nodes ${x_{{\mathcal{M}_c}}},c \in \mathcal{I}(d)$;}
\STATE {2)\; Each variable node ${x_{{\mathcal{M}_c}}}$ computes the mean $\mu _{{\mathcal{M}_{d,c}}}^{(\kappa )}$ and variance $\eta _{{\mathcal{M}_{d,c}}}^{(\kappa )}$ in (\ref{extrinsic_xm}), and sends them back to the connected factor nodes ${y_{{\mathcal{M}_d}}},d \in \mathcal{J}(c)$;}
\STATE {3)\; Calculate the convergence indicator $\alpha _\mathcal{M}^{(\kappa )}$ in (\ref{conv_m});}
\STATE {4)\; Update $P\left( {{{\bf{x}}_\mathcal{M}}} \right) = {{\bar P}^{(\kappa )}}\left( {{{\bf{x}}_\mathcal{M}}} \right)$ and ${P_E}\left( {{{\bf{x}}_\mathcal{M}}} \right) = \bar P_E^{(\kappa )}\left( {{{\bf{x}}_\mathcal{M}}} \right)$ if $\alpha _\mathcal{M}^{(\kappa )}>\alpha _\mathcal{M}^{(\kappa -1)}$;}
\STATE {5)\; $\kappa: = \kappa + 1$;}
\UNTIL{$\alpha _\mathcal{M}^{(\kappa )}=1$ or $\kappa={n_\mathcal{M}}$.}
\STATE {Output: ${P_E}\left( {{{\bf{e}}_v}} \right),\forall v \in {\cal V}$.}
\end{algorithmic}
\end{algorithm}

\textbf{1) From factor node ${y_{{\mathcal{M}_d}}}$ to variable nodes ${x_{{\mathcal{M}_c}}},c \in \mathcal{I}(d)$:} At each factor node, we can represent the received signal ${y_{{\mathcal{M}_d}}}$ as
\begin{align}\label{mobility_y}
{y_{{\mathcal{M}_d}}} = {H_{{\mathcal{M}_{d,c}}}}{x_{{\mathcal{M}_c}}} + \sum\limits_{e \in \mathcal{I}(d),e \ne c} {{H_{{\mathcal{M}_{d,e}}}}{x_{{\mathcal{M}_e}}}}  + {z_{{\mathcal{M}_d}}}.
\end{align}
The messages passed from the factor node ${y_{{\mathcal{M}_d}}}$ to variable node ${x_{{\mathcal{M}_c}}}$ are the mean $C_{{\mathcal{M}_{d,c}}}^{(\kappa )}$ and variance $D_{{\mathcal{M}_{d,c}}}^{(\kappa )}$, respectively, given by
\begin{align}\label{extrinsic_ym_m}
C_{{\mathcal{M}_{d,c}}}^{(\kappa )} = {{\left[ {{y_{{\mathcal{M}_d}}} - \sum\limits_{e \in \mathcal{I}(d),e \ne c} {{H_{{\mathcal{M}_{d,e}}}}\mu _{{\mathcal{M}_{d,e}}}^{(\kappa -1)}} } \right]} \mathord{\left/
 {\vphantom {{\left[ {{y_{{\mathcal{M}_d}}} - \sum\limits_{e \in \mathcal{I}(d),e \ne c} {{H_{{\mathcal{M}_{d,e}}}}\mu _{{\mathcal{M}_{d,e}}}^{(\kappa -1)}} } \right]} {{H_{{\mathcal{M}_{d,c}}}}}}} \right.
 \kern-\nulldelimiterspace} {{H_{{\mathcal{M}_{d,c}}}}}},
\end{align}
\begin{align}\label{extrinsic_ym_v}
D_{{\mathcal{M}_{d,c}}}^{(\kappa )} = {{\left[ {\sum\limits_{e \in \mathcal{I}(d),e \ne c} {{{\left| {{H_{{\mathcal{M}_{d,e}}}}} \right|}^2}\eta _{{\mathcal{M}_{d,e}}}^{(\kappa -1)}}  + {\sigma _{{\mathcal{M}_d}}}} \right]} \mathord{\left/
 {\vphantom {{\left[ {\sum\limits_{e \in \mathcal{I}(d),e \ne c} {{{\left| {{H_{{\mathcal{M}_{d,e}}}}} \right|}^2}\eta _{{\mathcal{M}_{d,e}}}^{(\kappa -1)}}  + {\sigma _{{\mathcal{M}_d}}}} \right]} {{{\left| {{H_{{\mathcal{M}_{d,c}}}}} \right|}^2}}}} \right.
 \kern-\nulldelimiterspace} {{{\left| {{H_{{\mathcal{M}_{d,c}}}}} \right|}^2}}},
\end{align}
where ${\mu _{{\mathcal{M}_{d,e}}}^{(\kappa -1)}}$ and ${\eta _{{\mathcal{M}_{d,e}}}^{(\kappa -1)}}$ are the mean and variance received from variable node ${x_{{\mathcal{M}_e}}}$ in the $(\kappa -1)$-th iteration. They can be initialized in the first iteration according to (\ref{Gaussi_pro}). ${{\sigma _{{\mathcal{M}_d}}}}$ is the $d$-th diagonal element of ${{\bf{\Sigma }}_\mathcal{M}}$.

\textbf{2) From variable node ${x_{{\mathcal{M}_c}}}$ to factor nodes ${y_{{\mathcal{M}_d}}},d \in \mathcal{J}(c)$:} The {\em a posteriori} probability at each variable node is given by
\begin{align}
{{\bar P}^{(\kappa )}}\left( {{x_{{\mathcal{M}_c}}} = \chi } \right) \propto {P_D}\left( {{x_{{\mathcal{M}_c}}} = \chi } \right)\mathop \prod \limits_{e \in \mathcal{J}(c)} \exp \left( { - \frac{{{{\left| {\chi  - C_{{\mathcal{M}_{e,c}}}^{(\kappa )}} \right|}^2}}}{{D_{{\mathcal{M}_{e,c}}}^{(\kappa )}}}} \right),\;\forall \chi  \in \mathbb{A}.
\end{align}
We again project this probability into a Gaussian distribution $\mathcal{CN}\left( {E_{{\mathcal{M}_c}}^{(\kappa )},F_{{\mathcal{M}_c}}^{(\kappa )}} \right)$ and set a minimum allowed variance $\varepsilon $, i.e., $F_{{\mathcal{M}_c}}^{(\kappa )} = \max \{ \varepsilon ,F_{{\mathcal{M}_c}}^{(\kappa )}\} $ to avoid numerical instabilities. We then update the extrinsic distribution $\bar q_E^{(\kappa )}\left( {{x_{{\mathcal{M}_{d,c}}}}} \right) \sim \mathcal{CN}\left( {\bar \mu _{{\mathcal{M}_{d,c}}}^{(\kappa )},\bar \eta _{{\mathcal{M}_{d,c}}}^{(\kappa )}} \right)$ in which 
\begin{align}
\bar \eta _{{\mathcal{M}_{d,c}}}^{(\kappa )} = {\left[ {{{\left( {F_{{\mathcal{M}_c}}^{(\kappa )}} \right)}^{ - 1}} - {{\left( {D_{{\mathcal{M}_{d,c}}}^{(\kappa )}} \right)}^{ - 1}}} \right]^{ - 1}}, \quad
\bar \mu _{{\mathcal{M}_{d,c}}}^{(\kappa )} = \bar \eta _{{\mathcal{M}_{d,c}}}^{(\kappa )}\left[ {\frac{{E_{{\mathcal{M}_c}}^{(\kappa )}}}{{F_{{\mathcal{M}_c}}^{(\kappa )}}} - \frac{{C_{{\mathcal{M}_{d,c}}}^{(\kappa )}}}{{D_{{\mathcal{M}_{d,c}}}^{(\kappa )}}}} \right].
\end{align}
Finally, the variable node ${x_{{\mathcal{M}_c}}}$ updates the mean $\mu _{{\mathcal{M}_{d,c}}}^{(\kappa )}$ and variance $\eta _{{\mathcal{M}_{d,c}}}^{(\kappa )}$ as follows and delivers them to the factor node ${y_{{\mathcal{M}_d}}}$.  
\begin{align}\label{extrinsic_xm}
\eta _{{\mathcal{M}_{d,c}}}^{(\kappa )} = {\left[ {\frac{{{\delta _\mathcal{M}}}}{{\bar \eta _{{\mathcal{M}_{d,c}}}^{(\kappa )}}} + \frac{{(1 - {\delta _\mathcal{M}})}}{{\eta _{{\mathcal{M}_{d,c}}}^{(\kappa -1)}}}} \right]^{ - 1}}, \quad
\mu _{{\mathcal{M}_{d,c}}}^{(\kappa )} = \eta _{{\mathcal{M}_{d,c}}}^{(\kappa )}\left[ {{\delta _\mathcal{M}}\frac{{\bar \mu _{{\mathcal{M}_{d,c}}}^{(\kappa )}}}{{\bar \eta _{{\mathcal{M}_{d,c}}}^{(\kappa )}}} + (1 - {\delta _\mathcal{M}})\frac{{\mu _{{\mathcal{M}_{d,c}}}^{(\kappa -1)}}}{{\eta _{{\mathcal{M}_{d,c}}}^{(\kappa -1)}}}} \right],
\end{align}
where ${\delta _\mathcal{M}} \in \left( {0,1} \right]$ is a damping factor applied to improve the accuracy and convergence. Similar to OAMP-LMMSE detector, we ignore the update if the variance $\eta _{{\mathcal{M}_{d,c}}}^{(\kappa )}$ is negative.

\textbf{3) Convergence indicator:} The convergence indicator $\alpha _\mathcal{M}^{(\kappa )}$ for mobile users is defined as
\begin{align}\label{conv_m}
\alpha _\mathcal{M}^{(\kappa )} = \frac{1}{{MN}}\sum\limits_{c = 0}^{MN - 1} {\mathbb{I}\left( {\mathop {\max }\limits_{\chi  \in \mathbb{A}} {{\bar P}^{(\kappa )}}\left( {{x_{{\mathcal{M}_c}}} = \chi } \right) \ge 1 - \varrho } \right)}.
\end{align}

\textbf{4) Update criterion:} If $\alpha _\mathcal{M}^{(\kappa )} > \alpha _\mathcal{M}^{(\kappa  - 1)}$, we update
\begin{align}
P\left( {{{\bf{x}}_\mathcal{M}}} \right) = {{\bar P}^{(\kappa )}}\left( {{{\bf{x}}_\mathcal{M}}} \right),\;{P_E}\left( {{{\bf{x}}_\mathcal{M}}} \right) = \bar P_E^{(\kappa )}\left( {{{\bf{x}}_\mathcal{M}}} \right),
\end{align}
where $\bar P_E^{(\kappa )}\left( {{x_{{\mathcal{M}_c}}} = \chi } \right) \propto \mathop \prod \limits_{e \in \mathcal{J}(c)} \exp \left( { - \frac{{{{\left| {\chi  - C_{{\mathcal{M}_{e,c}}}^{(\kappa )}} \right|}^2}}}{{D_{{\mathcal{M}_{e,c}}}^{(\kappa )}}}} \right),\forall \chi  \in \mathbb{A}$.

\textbf{5) Stopping criterion:} The GAMP-EP detector terminates when either $\alpha _\mathcal{M}^{(\kappa )} = 1$ or the maximum iteration number ${n_\mathcal{M}}$ is reached.

Once the stopping criterion is satisfied, we obtain the extrinsic probabilities ${P_E}\left( {{{\bf{e}}_v}} \right),\forall v \in {\cal V}$ according to the indices of each mobile user's symbols corresponding to ${P_E}\left( {{{\bf{x}}_{{\mathcal{M}}}}} \right)$.

\section{Performance Analysis and Complexity Reduction}\label{V_Reduced}
We now analyze the performance property of our proposed iterative SIC turbo receiver and develop reduced complexity variants for both the OAMP-LMMSE and GAMP-EP detectors.

\subsection{Performance Analysis with EXIT Chart}
Based on the main idea of EXIT chart\cite{li2005exit,el2013exit,lou2010soft}, we develop a novel customized variant to analyze the convergence behavior of our proposed iterative SIC turbo receiver.
An EXIT chart is a semi-analytical tool to study the transfer characteristics of mutual information (MI) 
between transmitted bits and their LLRs computed by receiver components through iterations in turbo detector.
It has been widely adopted for convergence behavior analysis and prediction of turbo processing in stationary communication systems involving
single user \cite{lou2010soft} or multiple users \cite{li2005exit}. 
Specifically, the receiver components are modeled as devices mapping the
{\em a priori} MI $I_i$ at the input to a new extrinsic MI $I_e$ at the 
output. Based on EXIT charts, extrinsic information exchanges
between the detector and the channel decoder can be visualized as a 
decoding trajectory. This yields an asymptotic convergence analysis 
for turbo receivers.

Unfortunately, EXIT chart cannot be directly applied to
convergence analysis of NOMA systems because
of user asymmetry in our proposed OBNOMA framework. 
To this end, we develop a customized EXIT chart for OBNOMA
to analyze the convergence of our proposed iterative SIC turbo receiver. 

Note that the input {\em a priori} information of mobile users 
will affect the output extrinsic information of stationary users in the 
SIC detector, and vice versa. 
Hence, we need to separately depict EXIT charts for stationary and mobile users. 
Specifically, 
we fix the input {\em a priori} MI of OBNOMA mobile users to
several different values (i.e., different $I_i^{(\mathcal{M})}$) 
in the SIC detector and generate the corresponding EXIT chart for OBNOMA stationary 
users. 
The system decoding trajectory path will follow the transfer 
curves of the channel decoders and the detectors based
on different $I_i^{(\mathcal{M})}$, and finally approach the desired operating point.

We can similarly generate EXIT chart for OBNOMA mobile users by fixing 
the input {\em a priori} MI of OBNOMA stationary users in the SIC detector. 
Compared with the traditional EXIT chart, the newly customized
EXIT chart can provide more insights into
the iterative behavior of the proposed SIC turbo receiver 
and graphically anticipate its convergence better. 

\begin{figure}[ht]
\begin{subfigure}{.5\textwidth}
  \centering
  % include first image
  \includegraphics[width=1\linewidth]{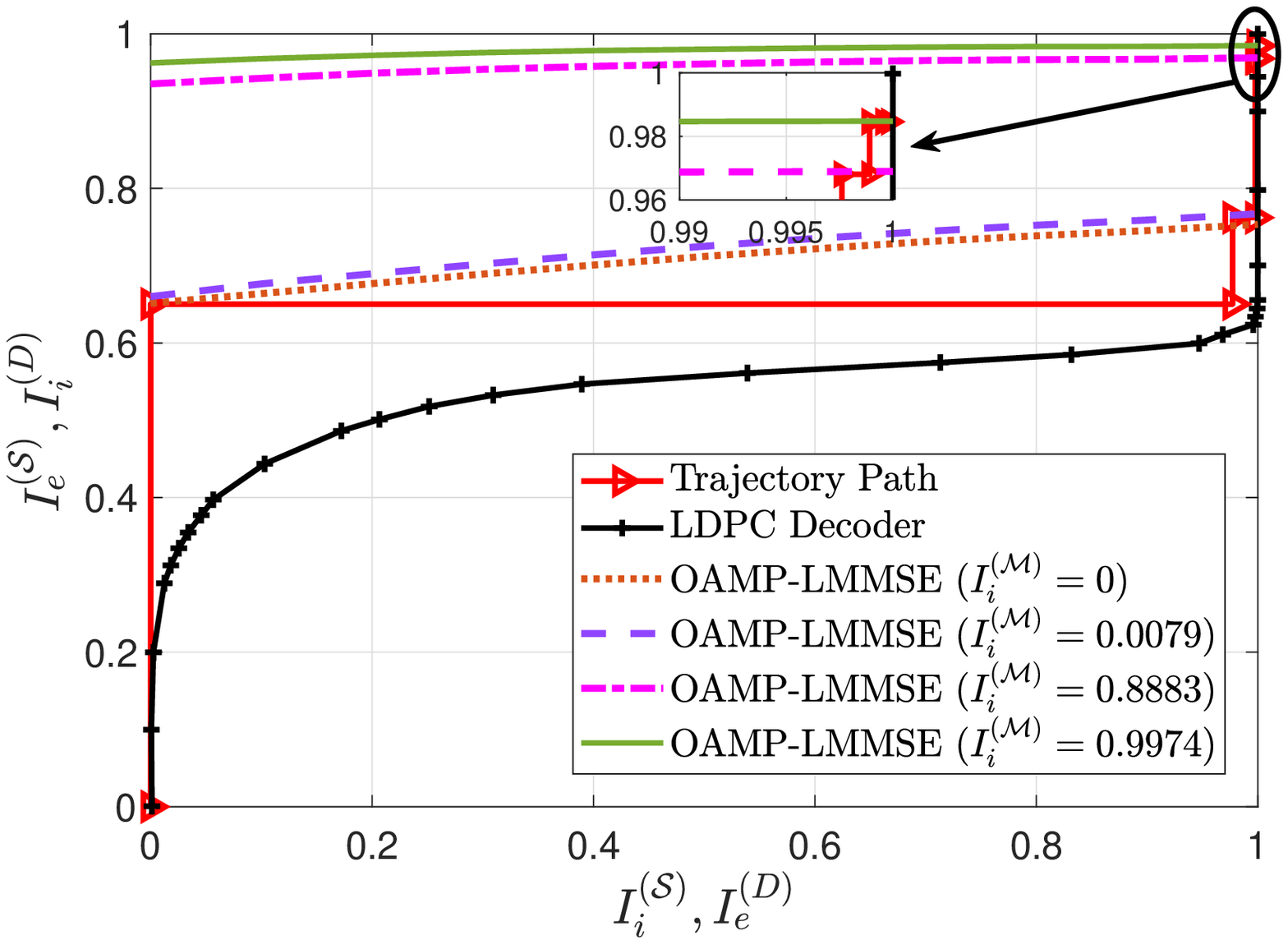}
  \caption{EXIT chart for stationary users.}
  \label{fig_EXIT_s}
\end{subfigure}
\begin{subfigure}{.5\textwidth}
  \centering
  % include second image
  \includegraphics[width=1\linewidth]{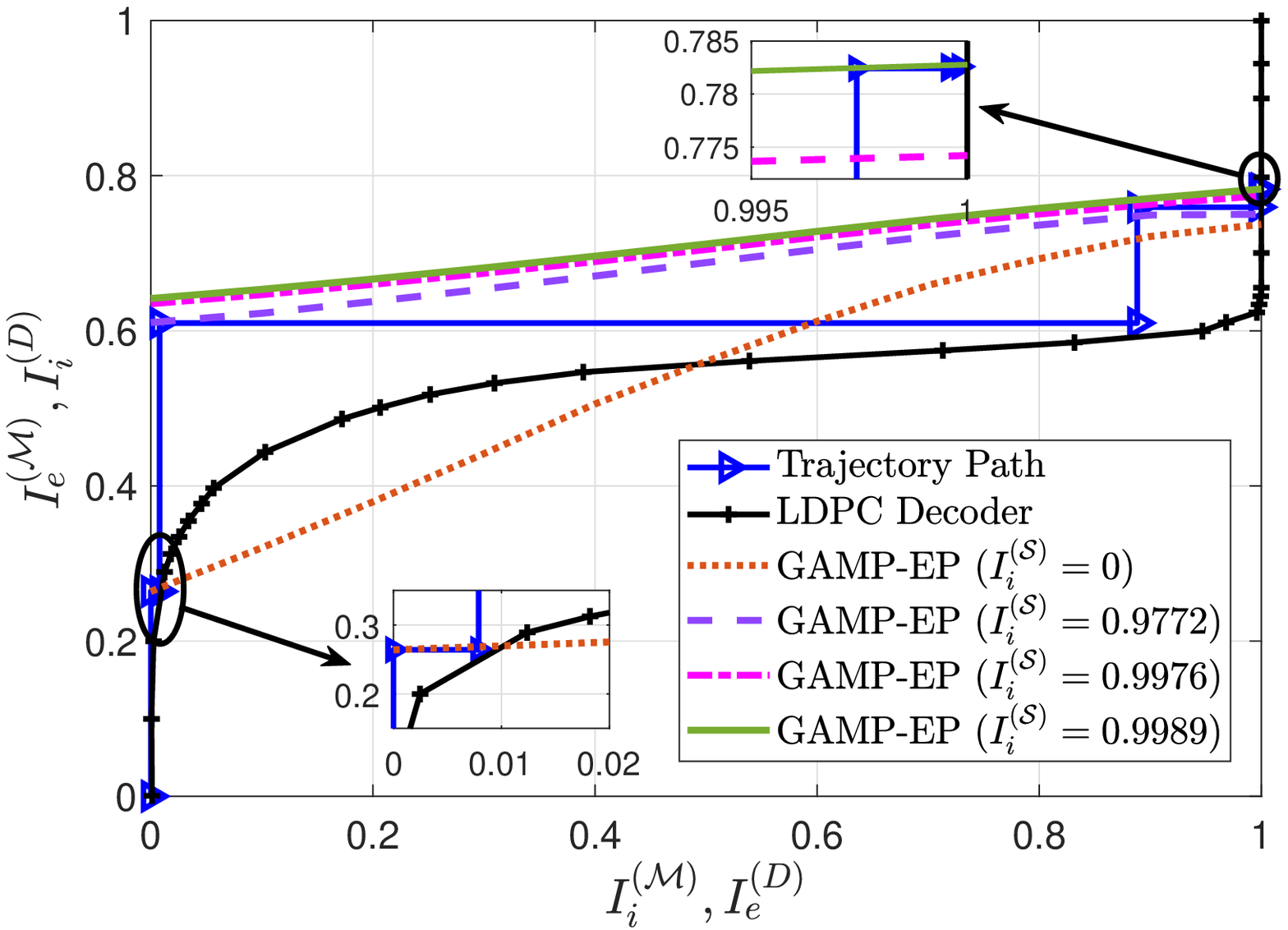}
  \caption{EXIT chart for mobile users.}
  \label{fig_EXIT_m}
\end{subfigure}
\caption{EXIT charts for the iterative SIC turbo receiver with ${E_{\cal S}}/{E_{\cal M}}=5$ dB and ${E_{\cal M}}/{N_0}=3.5$ dB.}
\label{fig_EXIT}
\end{figure}
To illustrate how the customized EXIT chart works for NBNOMA, 
Fig. \ref{fig_EXIT} shows an example of the proposed iterative SIC turbo receiver
in system with
$U=4$ stationary users and $V=4$ mobile users.
All user symbols are QPSK and channel decoders are low-density
parity-check (LDPC) as introduced in Section \ref{VI_Simulation}.
We also set the relative signal energy ratio between stationary
and mobile users to ${E_{\cal S}}/{E_{\cal M}}=5$ dB
and fix the mobile users' SNR to ${E_{\cal M}}/{N_0}=3.5$ dB. 
Here, a typical urban channel model \cite{failli1989digital} is adopted   
for each user and the channel response for each mobile user is generated
by utilizing Jakes formulation \cite{surabhi2019low,raviteja2018interference,ge2020} 
with maximum Doppler spread equals to $1111$ Hz. 
The stationary users' channels are generated with $0$ Doppler shift. 

From Fig. \ref{fig_EXIT_s} and Fig. \ref{fig_EXIT_m},
we observe that the system trajectories are staircase traces between the transfer curves of the detector and decoder components for both the stationary and mobile users.
In addition, the convergence region and average required number of iterations for the proposed iterative SIC turbo receiver can be predicted by EXIT charts.
By checking the quantity of staircase projections (steps) 
through the trajectory curves in Fig. \ref{fig_EXIT},
we notice that four iterations are already enough to achieve the expected performance. This convergence results corroborate the average
bit error rate (BER) of the receiver output for both stationary
and mobile users
in Fig. \ref{fig_EXIT_ite} with ${E_{\cal S}}/{E_{\cal M}}=5$ dB.
We found that the average BER drop becomes negligible beyond four iterations 
for both group of OBNOMA users. 
\begin{figure}
  \centering
  \includegraphics[width=3.4in]{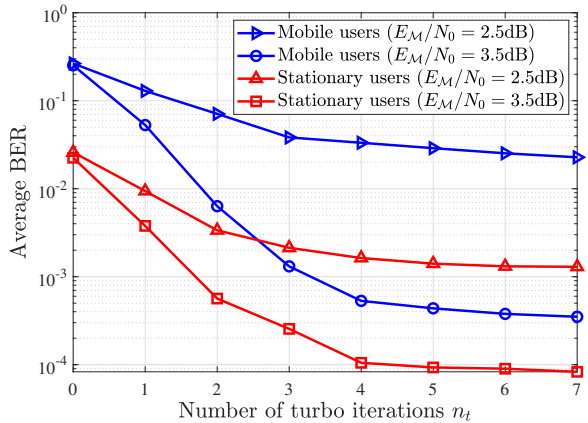}
  \caption{Average BER convergence of OBNOMA users under ${E_{\cal S}}/{E_{\cal M}}=5$ dB.}\label{fig_EXIT_ite}
\end{figure}

\subsection{Complexity Reduction}
The complexity of the proposed iterative SIC turbo receiver is mainly 
dominated by the advanced SIC detector.
TABLE~\ref{tab1} summarizes the implementation complexity 
for each iteration of OAMP-LMMSE, GAMP-EP and traditional MP algorithms \cite{raviteja2018interference,raviteja2019otfs,ramachandran2018mimo}.
The computational cost is measured according to the total number of 
real-field multiplications\footnote{Complex multiplication, inverse, and division 
account for three, four and six real-field multiplications, 
respectively.}, exponential functions, and matrix inverses, respectively.
\begin{table}%[h!]
  \begin{center}
    \caption{Complexity comparison of different algorithms for each iteration.}
    \label{tab1}
    \begin{tabular}{c|c|c|c} % <-- Alignments: 1st column left, 2nd middle and 3rd right, with vertical lines in between
    \hline
      Algorithm & Real-field Multiplication & Exponential & Matrix Inverse\\
      \hline
      OAMP-LMMSE & $6{M^2}N(M + 1) + 10MNQ + 34MN$ & $MNQ$ & $2N$\\
      GAMP-EP & $4QD' + 30D' + 2MNQ + 3MN$ & $QD'$ & -\\
      R-OAMP-LMMSE & $3{M^2}N(M + 3) + 10MNQ + 42MN$ & $MNQ$ & -\\
      R-GAMP-EP & $2MNR(2Q + 15) + 2MNQ + 3MN$ & $MNRQ$ & -\\
      MP & $22QD'+2D'+MNQ$ & $QD'$ & - \\
      \hline
    \end{tabular}
  \end{center}
\end{table}

We note that the OAMP-LMMSE detector complexity depends critically on 
matrix inverse.
The GAMP-EP detector complexity is related to the number of non-zero channel terms (i.e., $D'$) which represent channel matrix sparsity. 
Here, we write $\sum\limits_{d = 0}^{MN - 1} {\left| {{\cal I}(d)} \right|}  = \sum\limits_{c = 0}^{MN - 1} {\left| {{\cal J}(c)} \right|}  = D'$ for conciseness. 
However, the value $D'$ can sometimes be relatively large due to 
many off-grid channel delays and Doppler shifts. 
To reduce receiver complexity, we propose the corresponding low-complexity 
alternatives for OAMP-LMMSE and GAMP-EP, respectively.

\subsubsection{Reduced Complexity Algorithm for Stationary Users}
The complexity of OAMP-LMMSE for stationary users mainly arise from the matrix inverse in (\ref{S_Matrix_I}) with order $\mathcal{O}({M^3}N)$. As low-complexity approximations, similar to vector approximate message passing (VAMP) \cite{rangan2019vector}, we can use ${{\bar \sigma }_{{{\cal S}_k}}}{\bf{I}}$ and $\bar \eta _{{{\cal S}_k}}^{(\iota -1)}{\bf{I}}$ in place of ${{\bf{\Sigma }}_{{\mathcal{S}_k}}}$ and $\text{diag}{{\{ {\bm{\eta }}_{{\mathcal{S}_k}}^{(\iota -1)}\} }}$, respectively. The scalars ${{\bar \sigma }_{{{\cal S}_k}}}$ and $\bar \eta _{{{\cal S}_k}}^{(\iota -1)}$ are the sample average values of the diagonal elements of ${{\bf{\Sigma }}_{{\mathcal{S}_k}}}$ and the variance vector ${{\bm{\eta }}_{{\mathcal{S}_k}}^{(\iota -1)}}$, respectively. As a result, we can approximate (\ref{S_Matrix_I}) with
\begin{align}
{\bf{B}}_{{{\cal S}_k}}^{(\iota )} \approx {\bf{U}}_{{{\cal S}_k}}^H{\left( {{{({{\bar \sigma }_{{{\cal S}_k}}})}^{ - 1}}{\bf{\bar H}}_{{{\cal S}_k}}^H{{{\bf{\bar H}}}_{{{\cal S}_k}}} + {{\left(\bar \eta _{{{\cal S}_k}}^{(\iota -1)}\right)}^{ - 1}}{\bf{I}}} \right)^{ - 1}}{{\bf{U}}_{{{\cal S}_k}}},
\end{align}
where ${{{{\bf{\bar H}}}_{{{\cal S}_k}}}}$ and ${{\bf{U}}_{{{\cal S}_k}}}$ have the similar structures as ${{{\bf{\bar H}}}_k}$ in (\ref{H_diag}) and ${{\bf{U}}_k}$ in (\ref{Static_H}), respectively.
Since ${{({{\bar \sigma }_{{{\cal S}_k}}})}^{ - 1}}{\bf{\bar H}}_{{{\cal S}_k}}^H{{{\bf{\bar H}}}_{{{\cal S}_k}}} + {{\left(\bar \eta _{{{\cal S}_k}}^{(\iota -1)}\right)}^{ - 1}}{\bf{I}}$ is now diagonal, its inverse simply requires inverting the diagonal elements.

\subsubsection{Reduced Complexity Algorithm for Mobile Users}
As we can see, the channel factor graph in GAMP-EP 
has dense connections (edges), leading to relatively large value of $D'$. 
To reduce the resulting complexity, 
we adopt Gaussian approximation to trim part of these edges.

In particular, for each factor node ${y_{{\mathcal{M}_d}}}$, we would sort the corresponding $\left| {\mathcal{I}(d)} \right|$ channel coefficients according to their magnitudes, and choose $R$ largest terms to connect the corresponding edges in the factor graph while eliminating others. Towards this, we can rewrite the received signal ${y_{{\mathcal{M}_d}}}$ at $d$-th factor node in
(\ref{mobility_y}) as
%\begin{align}
%{y_{{{\cal M}_d}}} = \sum\limits_{e \in \Phi (d)} {{H_{{{\cal M}_{d,e}}}}{x_{{{\cal M}_e}}}}  + \underbrace {\sum\limits_{e \in \bar \Phi (d)} {{H_{{{\cal M}_{d,e}}}}{x_{{{\cal M}_e}}}}  + {z_{{{\cal M}_d}}}}_{z{'_{{{\cal M}_d}}}},
%\end{align}
\vspace{-0.4cm}
\begin{align}
{y_{{{\cal M}_d}}} = \sum\limits_{e \in \Phi (d)} {{H_{{{\cal M}_{d,e}}}}{x_{{{\cal M}_e}}}}  + \overbrace {\sum\limits_{e \in \bar \Phi (d)} {{H_{{{\cal M}_{d,e}}}}{x_{{{\cal M}_e}}}}  + {z_{{{\cal M}_d}}}}^{z{'_{{{\cal M}_d}}}},
\end{align}
where ${\Phi (d)}$ and ${\bar \Phi (d)}$ denote the index sets of the $R$ largest terms and the rest $(\left| {\mathcal{I}(d)} \right| - R)$ terms in $\mathcal{I}(d)$, respectively. The eliminated terms plus the noise can be approximately modeled as a Gaussian random variable ${z{'_{{{\cal M}_d}}}}$, where its mean and variance, respectively, given by 
\begin{align}
{\mu _{z{'_{{{\cal M}_d}}}}} = \sum\limits_{e \in \bar \Phi (d)} {{H_{{{\cal M}_{d,e}}}}{\mu _{{{\cal M}_e}}}}, \quad
{\sigma _{z{'_{{{\cal M}_d}}}}} = \sum\limits_{e \in \bar \Phi (d)} {{{\left| {{H_{{{\cal M}_{d,e}}}}} \right|}^2}{\eta _{{{\cal M}_e}}}}  + {\sigma _{{{\cal M}_d}}}.
\end{align}
Through this approach, the channel factor graph is simplified and only the dominant edges shall participate in message updates to approximate the true posterior in GAMP-EP detector.

To summarize, we include the complexity analyses of the proposed reduced complexity algorithms for OAMP-LMMSE and GAMP-EP
(denoted as R-OAMP-LMMSE and R-GAMP-EP, respectively) in TABLE \ref{tab1}
in comparison with the complexity of the original algorithms.

\section{Simulation Results}\label{VI_Simulation}
The performance of our proposed coded uplink OBNOMA scheme and iterative SIC turbo receiver are evaluated for different deployment scenarios in this section. In our simulation setups, we apply carrier frequency of $4$ GHz and subcarrier spacing ${\Delta f}=15$ kHz. Unless otherwise mentioned, we modulate the symbols by Gray-mapped QPSK and set the rolloff factor of the RRC filters as $0.4$ for both the transmitter and receiver. We generate a $(3,6)$-regular LDPC of length $2048$ with rate $1/2$ based on the progressive-edge growth (PEG) algorithm \cite{hu2005regular} and apply the belief propagation \cite{richardson2008modern} with a maximum number of $100$ iterations as channel decoder.

In OBNOMA, the delay-Doppler plane consists of $N=32$ and $M=128$. These delay-Doppler resources are allocated equally to $U=4$ stationary users and $V=4$ mobile users. A typical urban channel model \cite{failli1989digital} is applied with exponential power delay profile for both the stationary and mobile users. The velocity of the mobile user is set to ${\lambda _v} = 300$ km/h, resulting in a maximum Doppler spread $\nu _{{{\mathcal{M}_v}},\text{max}}=1111$ Hz, $\forall v \in \mathcal{V} $. For simplicity, we generate the Doppler shift for the $i$-th delay of the $v$-th mobile user by utilizing the Jakes formulation \cite{surabhi2019low,raviteja2018interference,ge2020}, i.e., $\nu _{{{\mathcal{M}_v}},i}=\nu _{{{\mathcal{M}_v}},\text{max}}\cos ({\rho _{{{\mathcal{M}_v}},i}})$, where ${\rho _{{{\mathcal{M}_v}},i}}$ is uniformly distributed over $[ - \pi ,\pi ]$.

We first assume that full CSI is available at the receiver and also study the impact of channel uncertainties on receiver performance. Without loss of generality, we set ${\delta _\mathcal{S}} ={\delta _\mathcal{M}}=0.3$, $\varepsilon ={10^{ - 8}}$, $\varrho=0.1$, ${n_\mathcal{S}}={n_\mathcal{M}}=20$ and choose $n_t=4$. These parameters were selected after extensive experimentations as a compromise between convergence speed and accuracy. All simulation results are averaged over 500 independent realizations.

\begin{figure}[ht]
\begin{subfigure}{.5\textwidth}
  \centering
  % include first image
  \includegraphics[width=1\linewidth]{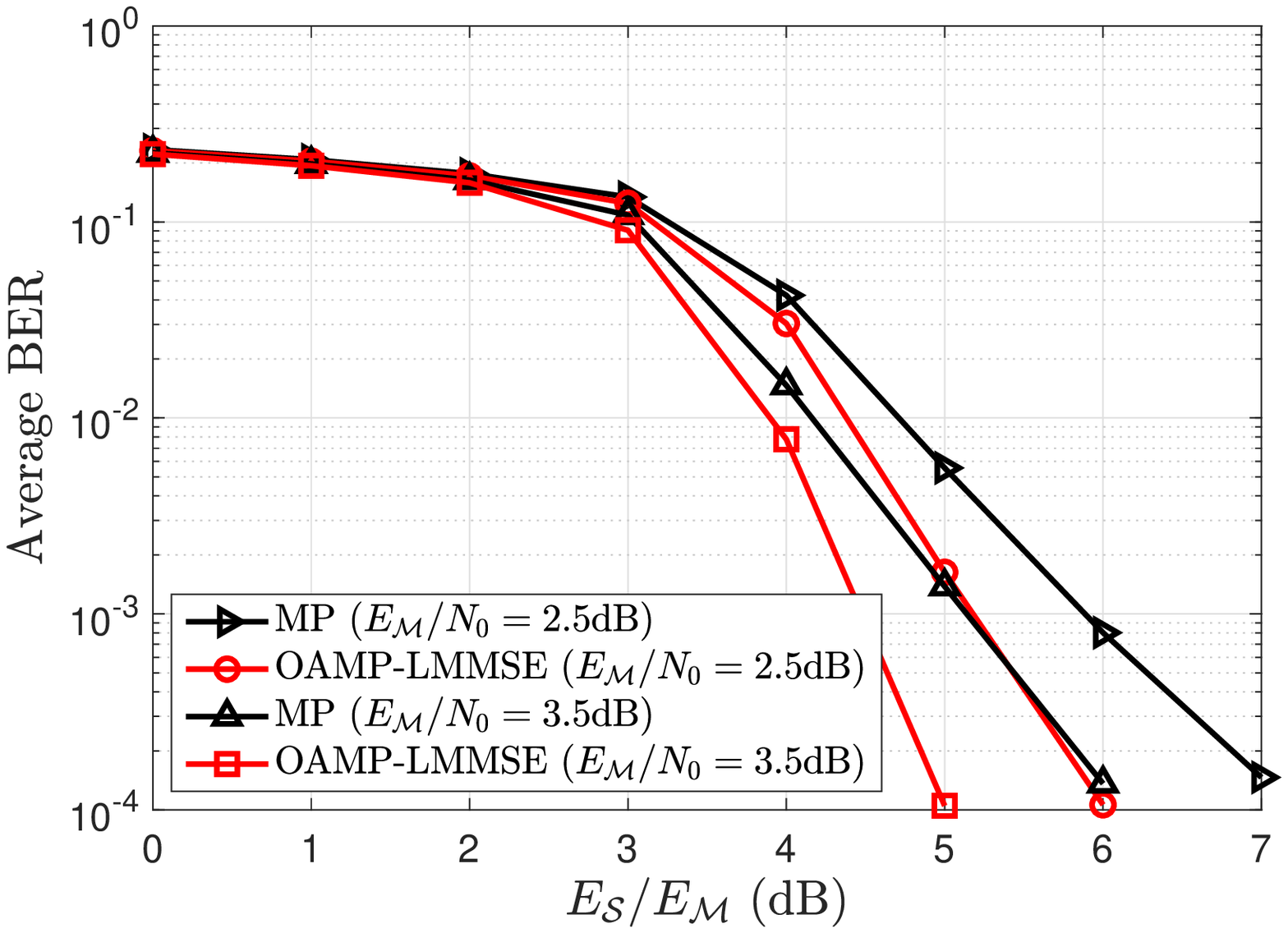}
  \caption{Average BER performance comparison for stationary users.}
  \label{fig:receiver_s}
\end{subfigure}
\begin{subfigure}{.5\textwidth}
  \centering
  % include second image
  \includegraphics[width=1\linewidth]{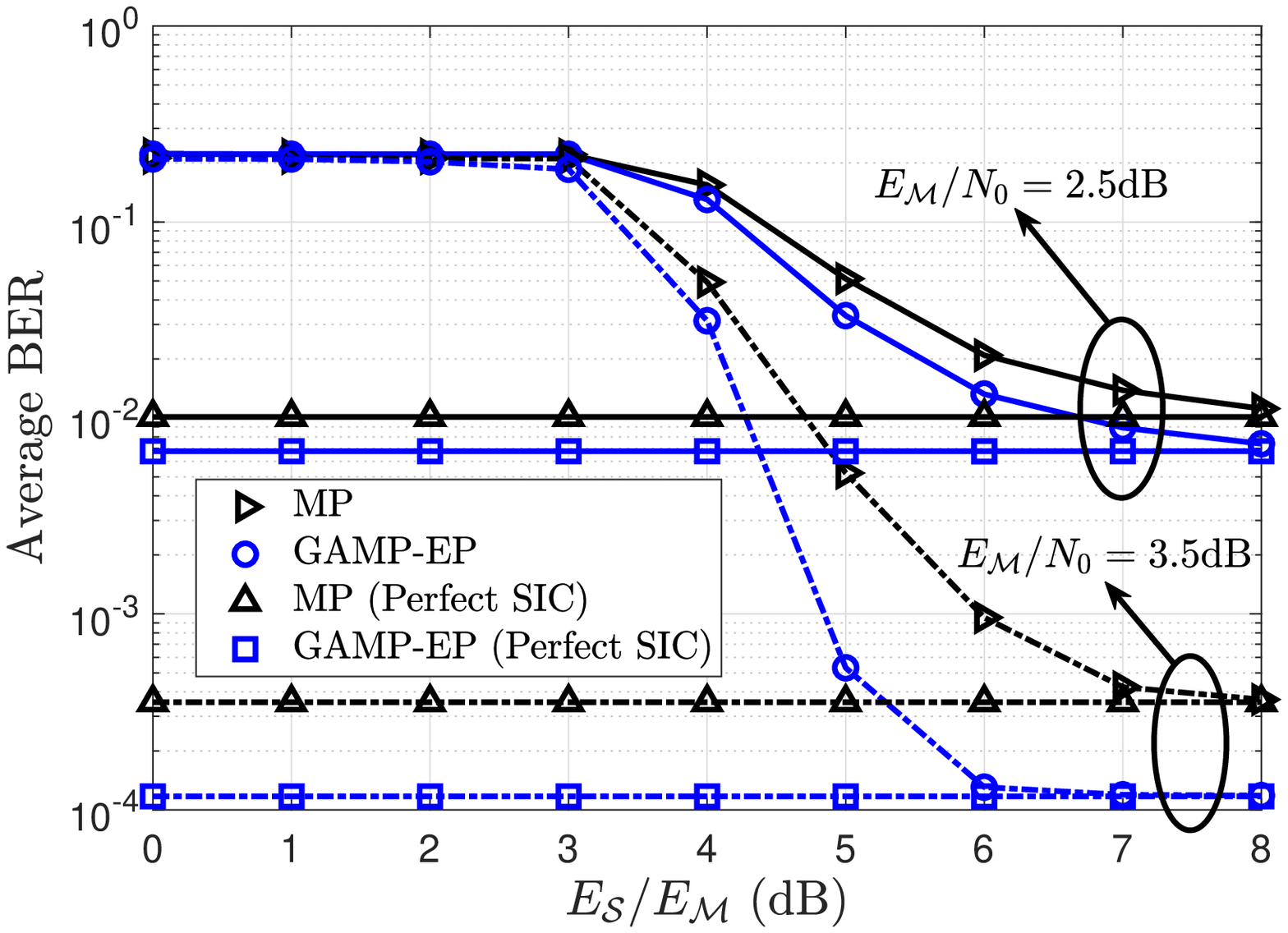}
  \caption{Average BER performance comparison for mobile users.}
  \label{fig:receiver_m}
\end{subfigure}
\caption{Average BER performance comparison of OBNOMA with different detector algorithms.}
\label{fig:receiver}
\end{figure}
Fig. \ref{fig:receiver} compares the average BER performance of the OBNOMA system with different detector algorithms for both the stationary users (Fig. \ref{fig:receiver_s}) and mobile users (Fig. \ref{fig:receiver_m}). To highlight the predominance of the proposed detecting algorithms, we also provide the performance of traditional MP algorithm \cite{raviteja2018interference,raviteja2019otfs,ramachandran2018mimo} as baselines for both the stationary and mobile users in Fig. \ref{fig:receiver}. 

The results reveal that all the receivers deliver improved performance with higher ${E_\mathcal{S}}/{E_\mathcal{M}}$ and ${E_\mathcal{M}}/{N_0}$. However, our proposed OAMP-LMMSE and GAMP-EP detectors outperform the MP detectors for both the stationary and mobile users. We also observe that as ${E_\mathcal{S}}/{E_\mathcal{M}}$ grows, the performance of both stationary and mobile users would improve. In particular, the mobile users' performance would asymptotically approach the performance of perfect SIC. Based on these analysis, we demonstrate that our proposed iterative SIC turbo receiver is practical and robustness against influence of imperfect SIC process.

Fig. \ref{Speed} shows the average BER performance of OBNOMA system at 
different mobile users' velocities with ${E_{\cal S}}/{E_{\cal M}}=5$ dB. 
As the velocities of the mobile users grow, we observe 
that the receiver performance improves modestly before saturation
for velocities beyond $450$ km/h. 
The underlying reason is that OTFS modulation can resolve high
contrast paths in the Doppler dimension at higher mobile users' velocities. 
Consequently, performance advantage becomes conceivable at 
higher user velocities (i.e., high Doppler spread channels).

We also test the average BER performance of OBNOMA when user symbols 
are modulated as 16-QAM. 
Fig. \ref{16QAM} illustrates the average BER performance 
of OBNOMA users with 16-QAM under ${E_{\cal M}}/{N_0}=7.5$ dB for
different settings of $M$ and $N$. 
We notice that the performance of 
both stationary and mobile users degrades for
smaller $M$ and $N$ due to loss of delay-Doppler grid
resolution. 
This results in the diversity loss as the receiver can only
resolve a smaller number of signal paths. 
These tests and results
strongly support the operational consistency of our proposed 
OBNOMA scheme and iterative SIC turbo receiver over different system parameters.
\begin{figure}
\begin{minipage}[t]{0.49\textwidth}%并排放两张图片，每张占页面的0.5，下同。
\centering
\includegraphics[width=3.3in]{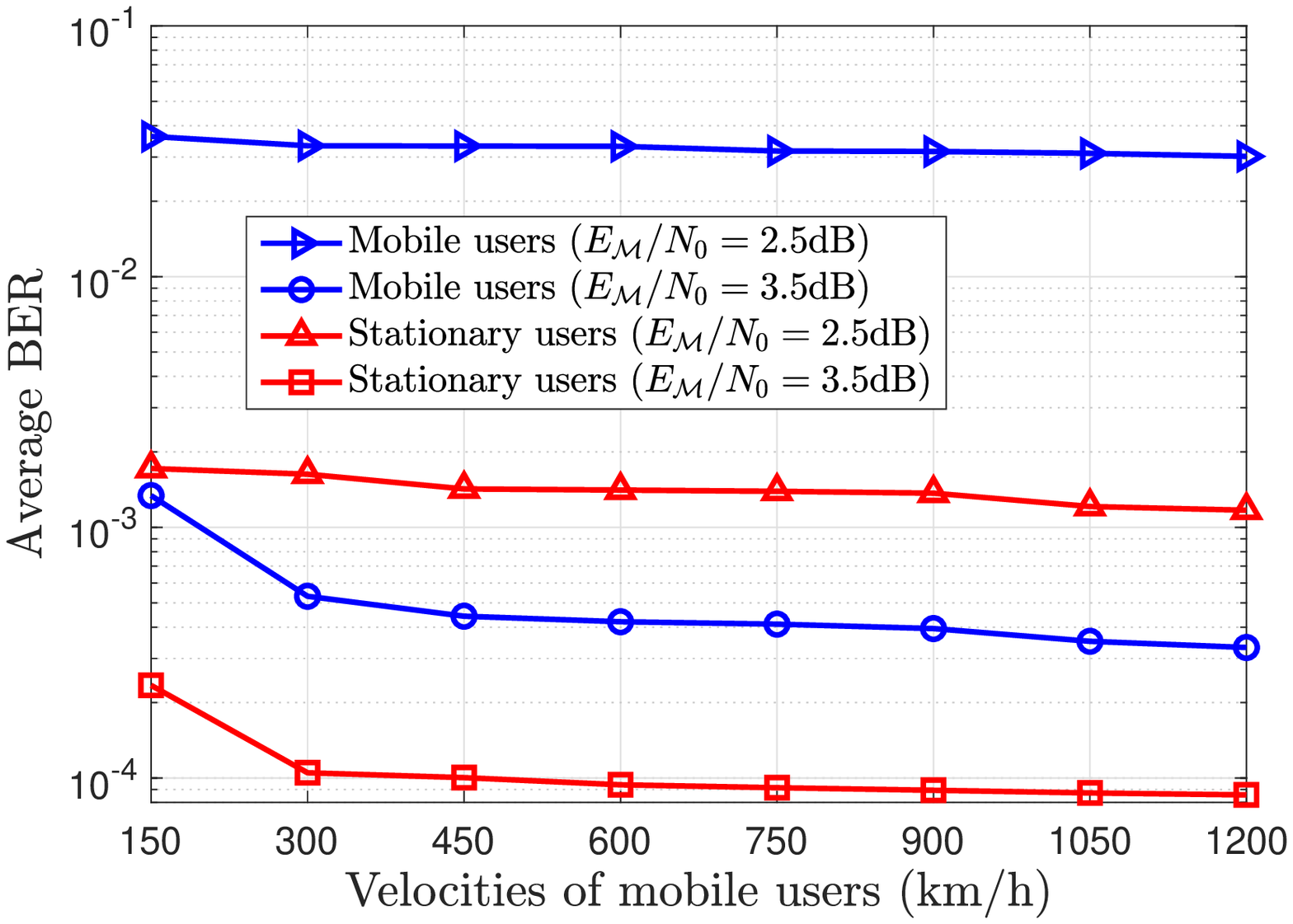}
\caption{Average BER performance of OBNOMA users at different mobile velocities with ${E_{\cal S}}/{E_{\cal M}}=5$ dB.}\label{Speed}
\end{minipage}
\hfill
\begin{minipage}[t]{0.49\textwidth}
\centering
\includegraphics[width=3.3in]{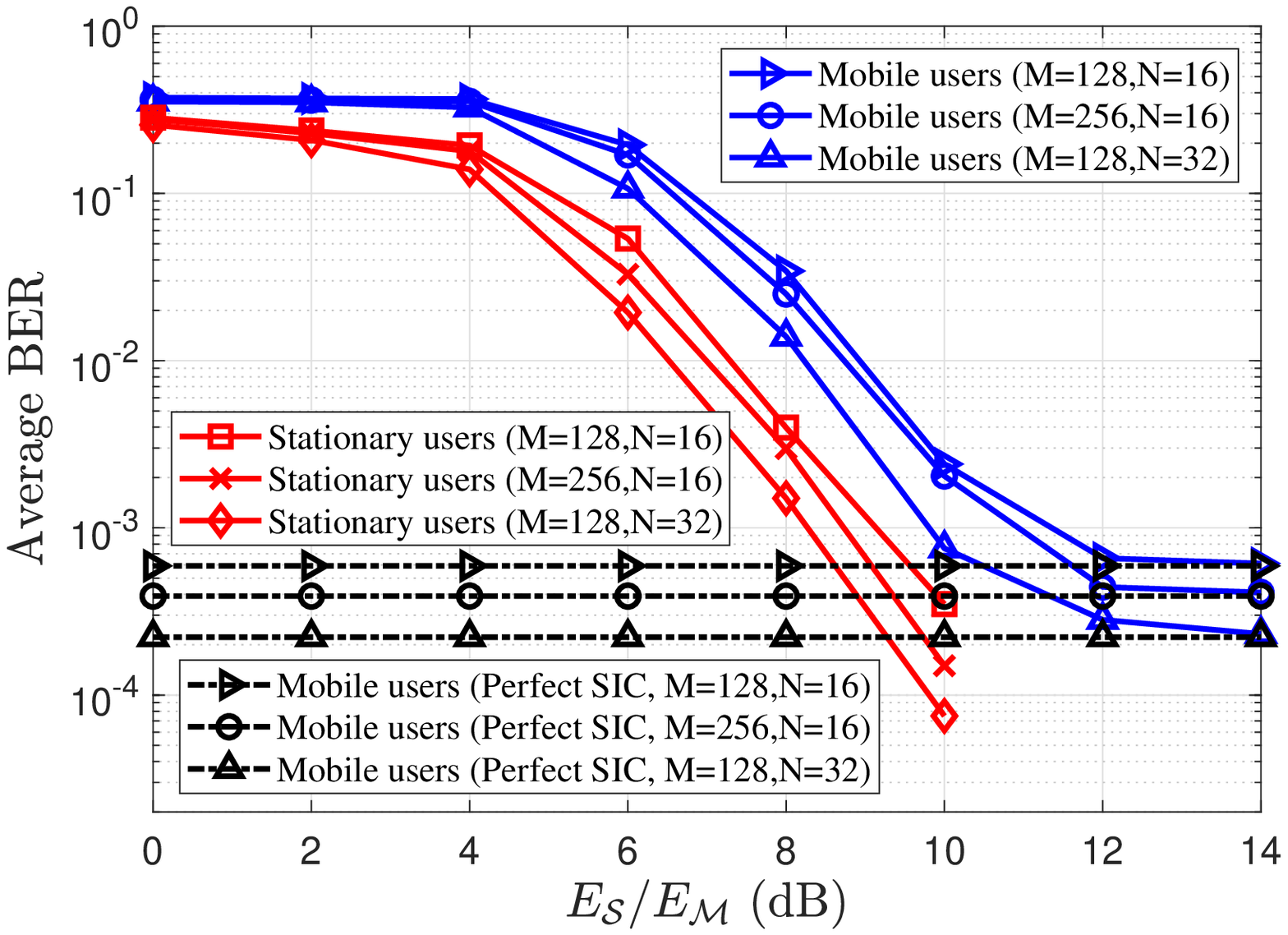}
\caption{Average BER performance of OBNOMA users utilizing 16-QAM symbols with ${E_{\cal M}}/{N_0}=7.5$ dB.}\label{16QAM}
\end{minipage}
\end{figure}

In terms of complexity reduction, Fig. \ref{fig:receiver_R} shows the average
BER performance of OBNOMA system 
for receivers utilizing the proposed reduced complexity detectors. 
In Fig. \ref{fig:receiver_R_s}, we observe that the R-OAMP-LMMSE detector 
achieves similar performance to that of OAMP-LMMSE detector for stationary users without costly matrix inverses, and achieves better performance than traditional MP detector. 
The results in Fig. \ref{fig:receiver_R_m} demonstrates a graceful
performance degradation for mobile users as the algorithm complexity
drops with smaller $R$. The results also reveal that as $R$ increases, the performance of R-GAMP-EP detector 
would asymptotically approach to that of full GAMP-EP detector and even better than that of traditional MP detector for mobile users. We also notice that the value of $R$ has slightly effect on the performance of stationary users.
Therefore, our proposed R-OAMP-LMMSE and R-GAMP-EP detectors can yield 
attractive compromise between receiver performance and complexity.
\begin{figure}[ht]
\begin{subfigure}{.5\textwidth}
  \centering
  % include first image
  \includegraphics[width=1\linewidth]{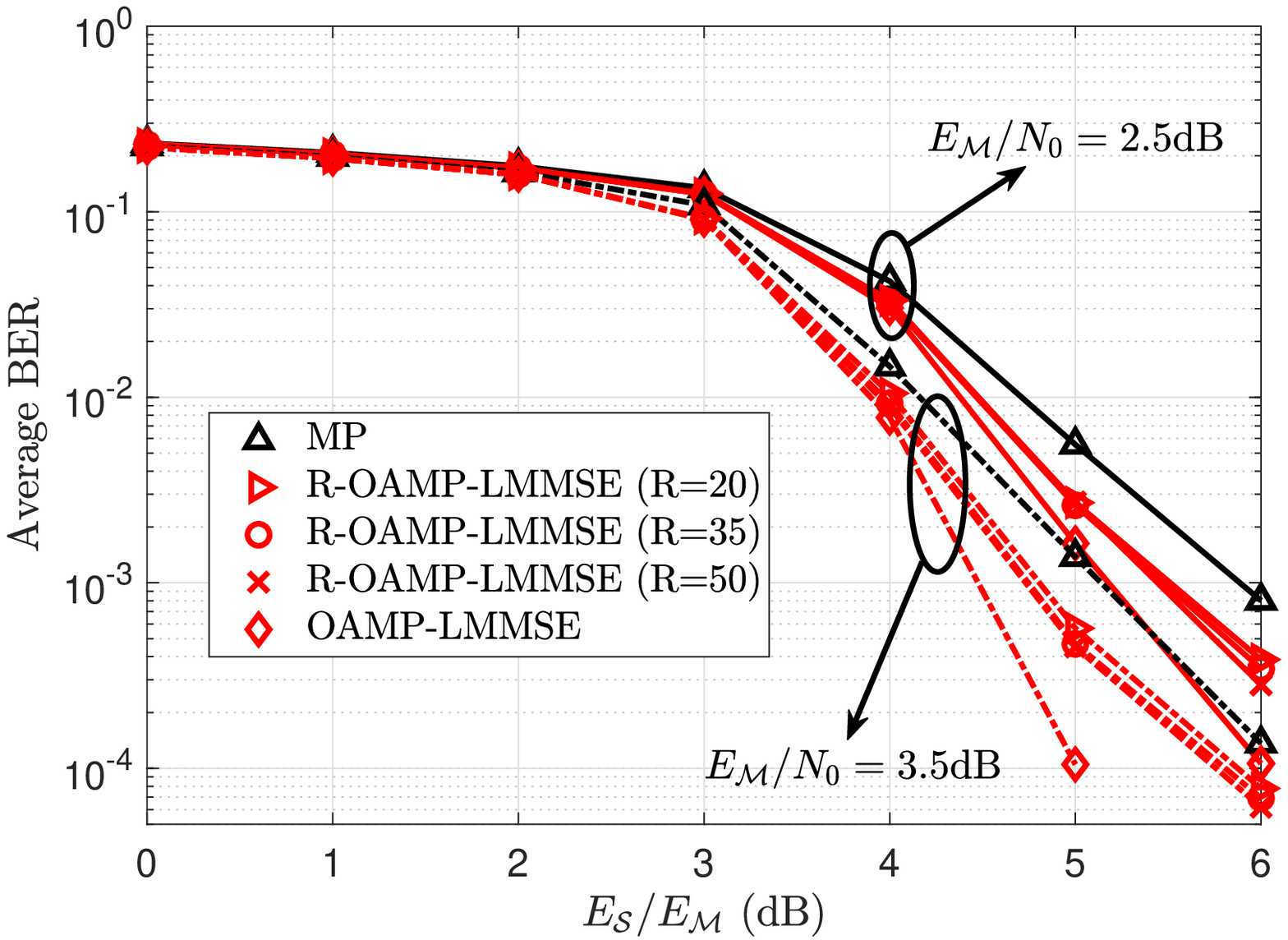}
  \caption{Average BER performance for stationary users.}
  \label{fig:receiver_R_s}
\end{subfigure}
\begin{subfigure}{.5\textwidth}
  \centering
  % include second image
  \includegraphics[width=1\linewidth]{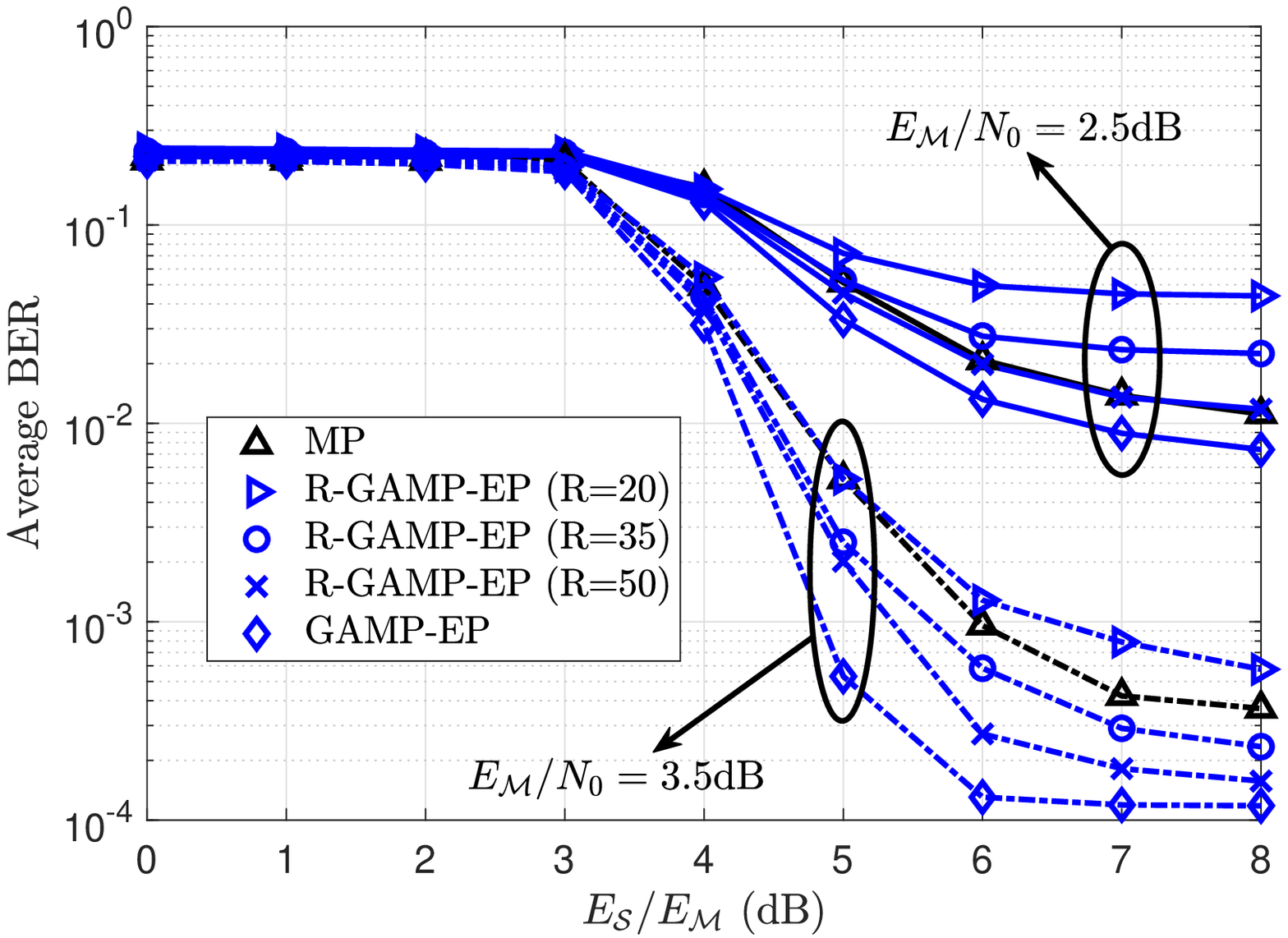}
  \caption{Average BER performance for mobile users.}
  \label{fig:receiver_R_m}
\end{subfigure}
\caption{Average BER performance of OBNOMA with reduced complexity detectors.}
\label{fig:receiver_R}
\end{figure}

Finally, the performance of the proposed iterative SIC turbo receiver are tested for imperfect CSI. Here, we characterize the channel uncertainties under norm-bounded CSI estimation errors, which can be modeled as
\begin{align*}
{h_{x,i}} &= {{\hat h}_{x,i}} + \Delta {h_{x,i}},\ \left\| {\Delta {h_{x,i}}} \right\| \le {\epsilon _{{h_{x,i}}}},\nonumber\\
{\tau _{x,i}} &= {{\hat \tau }_{x,i}} + \Delta {\tau _{x,i}},\ \left\| {\Delta {\tau _{x,i}}} \right\| \le {\epsilon _{{\tau _{x,i}}}},\nonumber\\
\nu _{{{\mathcal{M}_v}},i}&= {{\hat \nu }_{{{\mathcal{M}_v}},i}} + \Delta {\nu _{{{\mathcal{M}_v}},i}},\ \left\| {\Delta {\nu _{{{\mathcal{M}_v}},i}}} \right\| \le {\epsilon _{{\nu _{{{\mathcal{M}_v}},i}}}},\nonumber
\end{align*}
where $\forall x \in \left\{ {{\mathcal{S}_u},{\mathcal{M}_v}} \right\}$, ${{\hat h}_{x,i}}$, ${{\hat \tau }_{x,i}}$ and ${{\hat \nu }_{{{\mathcal{M}_v}},i}}$ denote the estimated values of ${h_{x,i}}$, ${\tau _{x,i}}$ and $\nu _{{{\mathcal{M}_v}},i}$. The corresponding channel estimation errors $\Delta {h_{x,i}}$, $\Delta {\tau _{x,i}}$ and $\Delta {\nu _{{{\mathcal{M}_v}},i}}$ are bounded in their norms. The model specifies the respective norm bounds of ${\epsilon _{{h_{x,i}}}}$, ${\epsilon _{{\tau _{x,i}}}}$ and ${\epsilon _{{\nu _{{{\mathcal{M}_v}},i}}}}$, respectively. For brevity, we assume that ${\epsilon _{{h_{x,i}}}} = \epsilon \left\| {{{\hat h}_{x,i}}} \right\|$,
${\epsilon _{{\tau _{x,i}}}} = \epsilon \left\| {{{\hat \tau }_{x,i}}}
 \right\|$ and ${\epsilon _{{\nu _{{{\mathcal{M}_v}},i}}}} = \epsilon \left\|
 {{\hat \nu }_{{{\mathcal{M}_v}},i}} \right\|,\forall u,v,i$.

From the results in Fig. \ref{CSI} under ${E_{\cal S}}/{E_{\cal M}}=5$ dB, we observe that the performance loss of our proposed schemes is mild for the modest values of channel uncertainty $\epsilon$. The performance of mobile users is more sensitive to the CSI uncertainty than the stationary users. 
The graceful degradation of receiver 
performance with increasing amount channel uncertainty demonstrate
the robustness of our proposed OBNOMA framework and
the iterative SIC turbo receiver against channel modeling errors.
\begin{figure}
  \centering
  \includegraphics[width=3.4in]{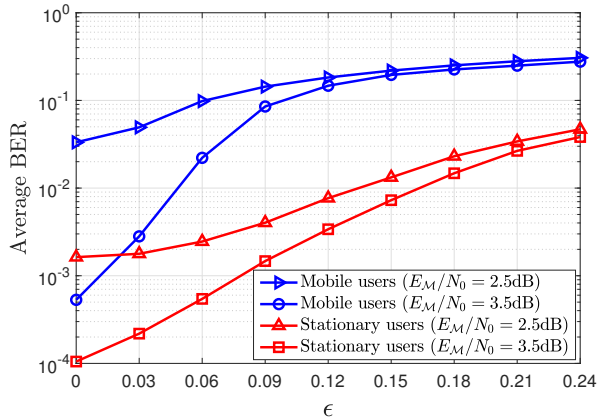}
  \caption{Average BER performance of OBNOMA users with imperfect CSI under ${E_{\cal S}}/{E_{\cal M}}=5$ dB.}\label{CSI}
\end{figure}

\section{Conclusion}\label{VII_Conclusion}
In this paper, we proposed a novel coded uplink multi-user system to achieve high spectrum efficiency
through NOMA. 
Our proposed OBNOMA framework groups
users with different mobility profiles for channel sharing.  
Based on the recently developed OTFS technology, 
we allocated the sub-vector resources of Doppler and delay dimensions 
respectively to the stationary and mobile users.
We developed an iterative SIC turbo receiver for 
effective multi-user detection and decoding under strong CCI of OBNOMA system.
We derived two detectors for different user mobility profiles
within the turbo receiver. We also proposed reduced complexity variants for both the detector algorithms without significant performance drop. Our EXIT chart analysis 
further verified the rapid convergence of the proposed 
receivers. 
Our results demonstrated the feasibility of OBNOMA as well as
strong performance and robustness of the proposed turbo receiver
against channel uncertainty and errors from imperfect SIC.

% or
\appendix % for no appendix heading
From (\ref{Static_H}), we have
\begin{align}
H_k[\ell,m]&=\sum_{r=0}^{M-1}\Lambda_k^H[\ell,\ell] F_M^H[\ell,r] \bar{H}_k[r,r]F_M[r,m]\Lambda_k[m,m]\notag \\
&=\frac{1}{M}\sum_{r=0}^{M-1} e^{j\frac{2\pi \ell k}{MN}} e^{j\frac{2\pi \ell r}{M}} H[k+rN]e^{-j\frac{2\pi rm}{M}}e^{-j\frac{2\pi km}{MN}}\notag \\
&=\frac{1}{M}\sum_{r=0}^{M-1} H[k+rN] e^{j\frac{2\pi (\ell-m) k}{MN}}e^{j\frac{2\pi (\ell-m) r}{M}}\notag \\
&=\frac{1}{M}\sum_{r=0}^{M-1} \sum_{p=0}^{P-1}h[p]e^{-j\frac{2\pi (k+rN)p}{MN}} e^{j\frac{2\pi (\ell-m) k}{MN}}e^{j\frac{2\pi (\ell-m) r}{M}}\notag \\
&=\frac{1}{M}\sum_{p=0}^{P-1}h[p]e^{j\frac{2\pi (\ell-m-p) k}{MN}} \sum_{r=0}^{M-1}e^{j\frac{2\pi (\ell-m-p) r}{M}}\notag \\
&=\sum_{p=0}^{P-1}h[p]e^{j\frac{2\pi (\ell-m-p) k}{MN}}\delta([\ell-m-p]_M).
\end{align}

Thus, using the definition of $\bar \gamma (k,\ell ,p)$ in (\ref{phase_static}), we have the following relationship from (\ref{relate_Static_equvalient})
\begin{align}
Y[\ell,k]&=\sum_{m=0}^{M-1}H_k[\ell,m]X[m,k]+\omega[\ell,k]\notag \\
&=\sum_{m=0}^{M-1}\sum_{p=0}^{P-1}h[p]e^{j\frac{2\pi (\ell-m-p) k}{MN}}\delta([\ell-m-p]_M)X[m,k]+\omega[\ell,k]\notag\\
&=\sum_{p=0}^{P-1}h[p]e^{j\frac{2\pi (\ell-[\ell-p]_M-p) k}{MN}}X[[\ell-p]_M,k]+\omega[\ell,k]\notag\\
&=\sum_{p=0}^{P-1}h[p]\bar \gamma (k,\ell ,p)X[[\ell-p]_M,k]+\omega[\ell,k],
\end{align}
which completes the proof.


%\appendices

%\section{}

%\section*{Acknowledgment}


\ifCLASSOPTIONcaptionsoff
  \newpage
\fi



%\bibliographystyle{IEEEtranTCOM}

%\begin{thebibliography}{1}


%\end{thebibliography}

\bibliographystyle{IEEEtran}
%\small
\footnotesize
\bibliography{ref_NOMA_OTFS}

%\begin{biography}[{\includegraphics[width=1in,height=1.25in,clip,keepaspectratio]{mshell}}]{Michael Shell}


%\vfill


\end{document}